\documentclass{imsart}
\usepackage{latexsym,epsfig,amssymb,amsmath,amsthm,color,url,bbm,pdfsync,tikz}
\usetikzlibrary{arrows,positioning} 
\usepackage{graphicx}
\usepackage{inputenc,todonotes}
\usepackage[square,comma,sort&compress]{natbib}
\setcitestyle{numbers,multirow,import}
\allowdisplaybreaks
\numberwithin{equation}{section}
\numberwithin{figure}{section}
\numberwithin{table}{section}

\newtheorem{Theorem}{Theorem}[section]
\newtheorem{Lemma}[Theorem]{Lemma}
\newtheorem{Remark}[Theorem]{Remark}
\newtheorem{Example}[Theorem]{Example}
\newtheorem{Proposition}[Theorem]{Proposition}
\newtheorem{Definition}[Theorem]{Definition}
\newtheorem{Corollary}[Theorem]{Corollary}

\def\btheta{\boldsymbol \theta}

\newcommand{\deltain}{\delta_{\text in}}
\newcommand{\deltaout}{\delta_{\text out}}

\newcommand{\bthe}{\begin{Theorem}}
\newcommand{\ethe}{\end{Theorem}}

\newcommand{\ble}{\begin{Lemma}}
\newcommand{\ele}{\end{Lemma}}

\newcommand{\bde}{\begin{Definition}}
\newcommand{\ede}{\end{Definition}}

\newcommand{\bco}{\begin{Corollary}}
\newcommand{\eco}{\end{Corollary}}

\newcommand{\bpr}{\begin{Proposition}}
\newcommand{\epr}{\end{Proposition}}

\newcommand{\brem}{\begin{Remark}}
\newcommand{\erem}{\end{Remark}}

\newcommand{\bexam}{\begin{Example}}
\newcommand{\eexam}{\end{Example}}

\newcommand{\beqq}{\begin{equation}}
\newcommand{\eeqq}{\end{equation}}

\newcommand{\beao}{\begin{eqnarray*}}
\newcommand{\eeao}{\end{eqnarray*}\noindent}

\newcommand{\beam}{\begin{eqnarray}}
\newcommand{\eeam}{\end{eqnarray}\noindent}

\newcommand{\barr}{\begin{array}}
\newcommand{\earr}{\end{array}}

\newcommand{\bproof}{\begin{proof}}
\newcommand{\eproof}{\end{proof}}

\usepackage[vlined, ruled]{algorithm2e}
\usepackage{default,float}
\SetArgSty{textnormal}
\SetCommentSty{textnormal}
\makeatletter
\newcommand{\algrule}[1][.2pt]{\par\vskip.5\baselineskip\hrule height #1\par\vskip.5\baselineskip}
\makeatother

\newcommand{\Din}{D_{\text{in}}}
\newcommand{\Dout}{D_{\text{out}}}

\newcommand{\Nin}{N^{\text{in}}}
\newcommand{\Nout}{N^{\text{out}}}
\newcommand{\pin}{p^{\text{in}}}
\newcommand{\pout}{p^{\text{out}}}
\newcommand{\Cin}{C_{\text{in}}}
\newcommand{\Cout}{C_{\text{out}}}
\newcommand{\ain}{\iota_\text{in}}
\newcommand{\aout}{\iota_\text{out}}
\newcommand{\din}{\delta_{\text{in}}}
\newcommand{\dout}{\delta_{\text{out}}}
\newcommand{\hatdin
}{\hat{\delta}_{\text{in}}}
\newcommand{\hatdout}{\hat{\delta}_{\text{out}}}

\newcommand{\tildedin}{\tilde{\delta}_{\text{in}}}
\newcommand{\tildedout}{\tilde{\delta}_{\text{out}}}
\newcommand{\pij}{p_{ij}}
\newcommand{\vstart}{v^{(1)}}
\newcommand{\vend}{v^{(2)}}

\newcommand{\PP}{\textbf{P}}
\newcommand{\EE}{\textbf{E}}
\newcommand{\ind}{\textbf{1}}
\newcommand{\argmax}{\operatornamewithlimits{argmax}}
\newcommand{\convp}{\stackrel{p}{\rightarrow}}
\newcommand{\convd}{\stackrel{d}{\rightarrow}}
\newcommand{\convas}{\stackrel{\text{a.s.}}{\longrightarrow}}

\newcommand{\Var}{\text{Var}}

\begin{document}
\bibliographystyle{plain}

\begin{frontmatter}

\title{Fitting the Linear Preferential Attachment Model}
\runtitle{Linear Preferential Attachment Model}

\begin{aug}
	
\author[1]{Phyllis Wan}\ead[label=e1]{phyllis@stat.columbia.edu},
\author[2]{Tiandong Wang}\ead[label=e2]{tw398@cornell.edu},
\author[1]{Richard A. Davis}\ead[label=e3]{rdavis@stat.columbia.edu},
\and
\author[2]{Sidney I.\ Resnick}\ead[label=e4]{sir1@cornell.edu}
\address[1]{Department of Statistics\\
Columbia University\\
1255 Amsterdam Avenue, MC 4690\\
New York, NY 10027\\
\printead{e1,e3}}
\address[2]{School of Operations Research and Information Engineering\\
Cornell University \\
Ithaca, NY 14853\\
\printead{e2,e4}}

\runauthor{Wan et al.}


\end{aug}


\begin{abstract} 

Preferential attachment is an {appealing mechanism} for modeling
power-law behavior of the degree distributions in directed social
networks. In this paper, we consider {methods for} fitting a 5-parameter linear
preferential model to network {data} under two data scenarios. In the case
where full history of the network formation is given, we derive the
maximum likelihood estimator of the parameters and show that it is
strongly consistent and asymptotically normal. In the case where only
a {single-time} snapshot of the network is available, we propose an estimation
method which combines method of moments with an approximation to the likelihood. The resulting estimator is also strongly consistent and
performs quite well compared to the MLE estimator. We illustrate both
estimation procedures through simulated data, and explore the usage of
this model in a real data example. 
\end{abstract}

\end{frontmatter}
\maketitle


\section{Introduction}\label{intro}

The preferential attachment mechanism, in which edges and nodes are
added to the network based on probabilistic rules, provides an
appealing description for the evolution of a network. The rule for how
edges connect nodes depends on node degree; {large degree nodes} 
attract more edges. The idea is applicable to both directed and
undirected graphs and is often the basis for studying social networks,
collaborator and citation networks, and recommender
networks. Elementary descriptions of the preferential attachment model
can be found in \cite{easley:kleinberg:2010} while more 
mathematical treatments are available in \cite{durrett:2010b,vanderHofstad:2017,bhamidi:2007}.
Also see \cite{kolaczyk:csardi:2014} for a statistical survey of
methods for network data, \cite{rinaldoEtAl2013} for
  consideration of statistics of an undirected network and
  \cite{yanEtal:2016} for asymptotics of a directed exponential
  random graph models. Limit theory for estimates of an undirected
  preferential attachment model was considered in \cite{GAO2017}.

{For many networks, empirical evidence} supports the hypothesis
that in- and out-degree distributions  follow a power law. This
property has been shown to hold in linear 
 preferential attachment models, which makes {preferential
   attachment} an attractive choice 
 for network modeling \cite{durrett:2010b,vanderHofstad:2017,
   krapivsky:2001,krapivsky:redner:2001,bollobas:borgs:chayes:riordan:2003}.  
While the marginal degree power laws in a simple linear preferential
attachment model were established in
\cite{krapivsky:2001,krapivsky:redner:2001,bollobas:borgs:chayes:riordan:2003},
the joint regular variation (see \cite{resnickbook:2008,
  resnickbook:2007})
 which is akin to a {\it joint power law}, was only recently
 established
 \cite{resnick:samorodnitsky:towsley:davis:willis:wan:2016, resnick:samorodnitsky:2015}.   
In addition, it was shown in \cite{wang:resnick:2016} that the joint probability mass function of the in- and out-degrees is multivariate regularly varying. This is a key result as the degrees of a network are integer-valued.

In this paper, we {discuss methods of} fitting a simple linear
preferential attachment model, which {is parametrized by}
$\boldsymbol{\theta} =(\alpha,\beta,\gamma,\din,\dout)$. The first
three parameters, $\alpha,\beta,\gamma$, correspond to probabilities of the 3
scenarios for adding an edge and hence sum to 1, i.e.,
$\alpha+\beta+\gamma=1$. The other two, $\din$ and $\dout$, are tuning
parameters related to growth rates. The tail indices of the marginal
power laws for the in- and out-degrees can be expressed as explicit
functions of $\boldsymbol{\theta}$ (see \eqref{c1} and \eqref{c2} below). The graph
$G(n)=(V(n),E(n))$, where $V(n)$ is the set of nodes and $E(n)$ is the
set of edges at the $n$th iteration, evolves  {based on postulates} that describe how new edges and nodes are formed. This construction of the network is Markov in the sense that the probabilistic rules for obtaining $G(n+1)$ once $G(n)$ is known do not require prior knowledge of earlier stages of the construction.

The Markov structure of the model allows us to construct a likelihood
function based on {observing}
$G(n_0),G(n_0+1),\ldots,G(n_0+n)$. After deriving the likelihood
function, we show that there exists a unique maximum at
$\hat{\boldsymbol{\theta}} =
(\hat\alpha,\hat\beta,\hat\gamma,\hatdin,\hatdout)$ and that the
resulting maximum likelihood estimator is strongly consistent and
asymptotically normal. The normality is proved using a martingale
central limit theorem applied to the score function. The limiting
distribution also reveals that $(\hat\alpha,\hat\beta,\hat\gamma)$,
$\hatdin$, and $\hatdout$ are asymptotically independent.  
From these results, asymptotic properties of the MLE for the power law indices can be derived.

For some network data, only a snapshot of the nodes and edges is
available at a single point in time, that is, only $G(n)$ is available
for some $n$. In such cases, we propose an estimation procedure for
the parameters of the network using an approximation to the
likelihood and method of moments. This also produces strongly consistent
estimators. {These estimators perform} reasonably well compared to the MLE where the entire evolution of the network is known {but predictably there is some loss of efficiency.}

We illustrate the estimation procedure for both scenarios {using}
simulated data. {S}imulation plays an important role in the process
of modeling networks {since} it provides a way to assess the
performance of model fitting procedures in the idealized setting of
knowing the true model. {Also}, after fitting a model to real
data, simulation provides a check on the quality of fit. Departures
from model assumptions can often be detected via simulation of
multiple realizations from the fitted network. Hence it is important
to have efficient simulation algorithms for producing realizations of
the preferential attachment network for a given set of parameter
values.  
We adopt a simulation method, learned from Joyjit Roy, that was {inspired} by \cite{atwood:2015} and is similar to that of \cite{tonelli:concas:locci:2010}.

{Our fitting methods are implemented} in a real data setting using
the Dutch Wiki talk network {\cite{kunegis:2013}.} While one should not
expect the simple 
5-parameter (later extended to 7-parameter) linear preferential
attachment model to fully explain a network with millions of edges, it
does provide a reasonable fit to the tail behavior of the degree
distributions. We are also able to detect important structural
features in the network through 
{fitting the model over} separate time intervals. 

{Often it is difficult to believe in the existence of a true model, especially one whose parameters remain constant over time.  Allowing, as we do, a preferential attachment model with only a few parameters and no possibility for node removal may seem simplistic and unrealistic for social network data.  Of course, preferential attachment is only one mechanism for network formation and evidence for its use in fields outside data networks is mixed \cite{jones:handcock:2003, jones2003social}  and we restrict attention to linear preferential attachment.  Even imperfect models have the potential to capture salient properties in the data, such as heavy-tailedness of the in-degree and out-degree distributions, and to identify departures from model assumptions.  While maximum likelihood estimation is essentially the gold standard for cases when the underlying model is a good representation of the data, it may perform poorly in case the model is far from being appropriate.  In forthcoming work, we consider a semi-parametric estimation approach for network models that exhibit heavy-tailed degree distributions.  This alternative estimation methodology borrows ideas from extreme value theory.   }

The rest of the paper is structured as follows. In Section \ref{sec:sim}, we formulate the linear
preferential attachment network model and present an efficient simulation method for the network.  Section \ref{sec:estMLE} gives parameter estimators  when either the full history is known or when only a single snapshot {in time} is available. We test these estimators against simulated data in Section \ref{sec:estSim}  and then explore the Wiki talk network in Section \ref{sec:estReal}. 


\section{Model specification and simulation}\label{sec:sim}

In this section, we present the linear preferential attachment model in detail and provide a fast simulation algorithm for the network.

\subsection{The linear preferential attachment model}\label{subsec:linpref}
The directed edge  preferential attachment model
\cite{bollobas:borgs:chayes:riordan:2003,krapivsky:redner:2001}
constructs  a growing directed random graph $G(n)=(V(n),E(n))$ whose 
dynamics depend on five non-negative real numbers
$\alpha, 
\beta, \gamma$, $\delta_{\text in}$ and $\delta_{\text out}$, where $\alpha+\beta+\gamma=1$ and $\din,\dout >0$. To avoid degenerate situations,
assume  that each of the numbers $\alpha, 
\beta, \gamma$ is strictly smaller than 1. 
We obtain a new graph $G({n})$ 
by adding  one edge to the existing graph $G({n-1})$ and 
index the constructed graphs by the number $n$ of edges in $E(n)$.
We start with an arbitrary initial finite directed
graph $G({n_0})$ with at least one node and $n_0$ edges.
For $n >n_0$,
$G(n)=(V(n),E(n))$ is a graph with $|E(n)|=n$ edges and a random number $|V(n)|=N(n)$ of
nodes. If $u\in V(n)$, $D_{\rm in}^{(n)}(u)$ and $D_{\rm out}^{(n)}(u)$ 
denote the in- and out-degree of $u$ respectively in $G(n)$. 
There are three
scenarios that we call the $\alpha$, $\beta$ and 
$\gamma$-schemes, which are activated by flipping a
3-sided coin whose outcomes are $1,2,3$ with probabilities
$\alpha,\beta,\gamma$. More formally,  we have an iid sequence
 of multinomial random
variables $\{J_n, n>n_0\}$ with cells labelled $1,2,3$ and cell
probabilities $\alpha,\beta,\gamma$. 
Then the graph $G(n)$ is
obtained from  $G(n-1)$ as follows.

\tikzset{
    >=stealth',
    punkt/.style={
           rectangle,
           rounded corners,
           draw=black, very thick,
           text width=6.5em,
           minimum height=2em,
           text centered},
    pil/.style={
           ->,
           thick,
           shorten <=2pt,
           shorten >=2pt,}
}
\newsavebox{\mytikzpic}
\begin{lrbox}{\mytikzpic} 
     \begin{tikzpicture}
    \begin{scope}[xshift=0cm,yshift=1cm]
      \node[draw,circle,fill=white] (s1) at (2,0) {$v$};
      \node[draw,circle,fill=gray!30!white] (s2) at (.5,-1.5) {$w$};
      \draw[->] (s1.south west)--(s2.north east){};
      \draw[dashed] (0,-2.2) circle [x radius=2cm, y radius=15mm];
    \end{scope}
    
     \begin{scope}[xshift=5cm,yshift=1cm]
      \node[draw,circle,fill=gray!30!white] (s1) at (.5,-1.5) {$v$};
      \node[draw,circle,fill=gray!30!white] (s2) at (-.5,-2.5) {$w$};
      \draw[->] (s1.south west)--(s2.north east){};
      \draw[dashed] (0,-2.2) circle [x radius=2cm, y radius=15mm];
    \end{scope}
    
     \begin{scope}[xshift=10cm,yshift=1cm]
      \node[draw,circle,fill=white] (s1) at (2,0) {$w$};
      \node[draw,circle,fill=gray!30!white] (s2) at (.5,-1.5) {$v$};
      \draw[->] (s2.north east)--(s1.south west){};
      \draw[dashed] (0,-2.2) circle [x radius=2cm, y radius=15mm];
    \end{scope}
    
     \node at (0,-3.5) {$\alpha$-scheme};
     \node at (5,-3.5) {$\beta$-scheme};
     \node at (10,-3.5) {$\gamma$-scheme};
  \end{tikzpicture}    

\end{lrbox}
  \begin{figure}[h]
    \centering 
    \usebox{\mytikzpic} 
\end{figure} 

\begin{itemize}
\item 
If $J_n=1$ (with probability
$\alpha$),  append to $G(n-1)$ a new node $v\in V(n)\setminus V(n-1)$ and an edge
$(v,w)$ leading
from $v$ to an existing node $w \in V(n-1)$.
Choose the existing node $w\in V(n-1)$  with probability depending
on its in-degree in $G(n-1)$:
\beqq \label{eq:probIn}
\PP[\text{choose $w\in V(n-1)$}] = \frac{D_{\rm in}^{(n-1)}(w)+\delta_{\text
    in}}{n-1+\delta_{\text in}N(n-1)} \,.
\eeqq
\item If $J_n=2$ (with probability $\beta$), add a directed edge
$(v,w) $ to $E({n-1})$ with $v\in V(n-1)=V(n) $ and $w\in V(n-1)=V(n) $ and 
 the existing nodes $v,w$ are chosen independently from the nodes of $G(n-1)$ with
 probabilities 
\beqq \label{eq:probBoth}
\PP[\text{choose $(v,w)$}] = \Bigl(\frac{D_{\rm out}^{(n-1)}(v)+\delta_{\text
    out}}{n-1+\delta_{\text out}N(n-1)}\Bigr)\Bigl(
 \frac{D_{\rm in}^{(n-1)}(w)+\delta_{\text
    in}}{n-1+\delta_{\text in}N(n-1)}\Bigr).
\eeqq
\item  If $J_n=3$ (with probability
$\gamma$),  append to $G(n-1)$ a new node $w\in V(n)\setminus V(n-1)$ and an edge $(v,w)$ leading
from  the existing node $v\in V(n-1)$  to the new node $w$. Choose
the existing node $v\in V(n-1)$ with probability
\beqq \label{eq:probOut}
\PP[\text{choose $v \in V(n-1)$}] = \frac{D_{\rm out}^{(n-1)}(v)+\delta_{\text
    out}}{n-1+\delta_{\text out}N(n-1)}\,.
\eeqq
\end{itemize}
{
Note that this construction allows the possibility of having self loops in the case where $J_n=2$, 
but the proportion of edges that are self loops goes to 0 as $n\to\infty$. Also, multiple edges are allowed between two nodes.}


{
\subsection{Power law of degree distributions}

Given an observed network with $n$ edges, let $N_{ij}(n)$ denote the number of nodes with in-degree $i$ and out-degree $j$. If the network is generated from the linear preferential attachment model described above, then from \cite{bollobas:borgs:chayes:riordan:2003}, 
  there exists a proper probability distribution $\{f_{ij}\}$ such
  that almost surely
\beqq\label{pij}
\frac{N_{ij}(n) }{N(n)} \to f_{ij}=:\frac{p_{ij}}{1-\beta},\quad n\to\infty.
\eeqq
Consider the limiting marginal in-degree distribution $\pin_i:=\sum_j p_{ij}$.
It is calculated from \cite[Equation (3.10)]{bollobas:borgs:chayes:riordan:2003} that
\begin{align*}
\pin_0 &= \frac{\alpha}{1+a_1(\din)\din},\\
\pin_i &= \frac{\Gamma(i+\din)\Gamma(1+\din+a_1(\din)^{-1})}{\Gamma(i+1+\din+a_1(\din)^{-1})\Gamma(1+\din)}\left(\frac{\alpha\din}{1+a_1(\din)\din}+\frac{\gamma}{a_1(\din)}\right),\quad i\ge 1,
\end{align*}
where 
$$a_1(\lambda) := \frac{\alpha + \beta}{1+\lambda(1-\beta)}, \quad \lambda>0.$$
Moreover, $\pin_i$ satisfies
\begin{align}
\pin_i &:= \sum_{j=0}^\infty \pij \sim \Cin i^{-\ain}\mbox{ as }i\to\infty, \quad\text{as long as }\alpha\din+\gamma>0,\label{asyI}
\end{align}
for some finite positive constant $\Cin$,
where the power index
\beqq\label{c1}
\ain = 1+\frac{1+\din(\alpha+\gamma)}{\alpha+\beta}
\eeqq
Similarly, the limiting marginal out-degree distribution has the same property:
\begin{align*}
\pout_j &:= \sum_{i=0}^\infty \pij \sim \Cout i^{-\aout}\mbox{ as }j\to\infty, \quad\text{as long as }\gamma\dout+\alpha>0,
\end{align*}
for $\Cout$ positive and 
\beqq\label{c2}
\aout = 1+\frac{1+\dout(\alpha+\gamma)}{\beta+\gamma}.
\eeqq

}

\subsection{Simulation algorithm}\label{subsec:sim}
\begin{algorithm}[t]
\label{algo:1}
\SetAlgoLined\DontPrintSemicolon
\SetKwProg{myalg}{Algorithm}{}{}
\myalg{}{
	\vspace{.05in}
	\KwIn{$\alpha,\beta,\deltain,\deltaout$, the parameter values; $G(n_0) = (V(n_0),E(n_0))$, the initialization graph; $n$, the targeted number edges}
	\KwOut{$G(n) = (V(n),E(n))$, the resulted graph}
	\vspace{.05in}
 	$t\gets n_0$\;
 	\While{$t < n$}{
 		$N(t) \gets |V(t)|$\;
 		Generate $U\sim Uniform(0,1)$\;
		\uIf{$U<\alpha$}{
			$v^{(1)} \gets N(t)+1$\;
			$v^{(2)} \gets $\textsf{ Node\_Sample}$(E(t),2,\deltain)$\;
			$V(t) \gets {\textsf{Append}}(V(t),N(t)+1)$\;
		}
		\uElseIf{$\alpha\le U<\alpha+\beta$}{
			$v^{(1)} \gets $\textsf{ Node\_Sample}$(E(t),1,\deltaout)$\;
			$v^{(2)} \gets $\textsf{ Node\_Sample}$(E(t),2,\deltain)$\;
		}
		\uElseIf{$U\ge\alpha+\beta$}{
			$v^{(1)} \gets $\textsf{ Node\_Sample}$(E(t),1,\deltaout)$\;
			$v^{(2)} \gets N(t)+1$\;
			$V(t) \gets {\textsf{Append}}(V(t),N(t)+1)$\;
		}
		$E(t+1) \gets $ {\textsf{Append}}$(E(t), (v^{(1)},v^{(2)}))$\;
		$t \gets t+1$\;
	 }
\Return{$G(n)=(V(n),E(n))$}\;
}

\algrule[.8pt]
\setcounter{AlgoLine}{1}
\SetKwProg{myproc}{Function}{}{}
\myproc{\textsf{ Node\_Sample}}{
	\vspace{.05in}
  	\DontPrintSemicolon
 	\KwIn{$E(t)$, the edge list up to time $t$; $j=1,2$, the node to be sample, representing outgoing and incoming nodes, respectively;  $\delta \in \{\deltain,\deltaout\}$, the offset parameter}
 	\KwOut{the sampled node, $v$}
	Generate $W\sim Uniform(0,t+N(t)\delta)$\;
		\uIf{$W\le t$}{
			$v \gets v_{\lceil W\rceil}^{(j)}$\;
		}
		\uElseIf{$W > t$}{
			$v \gets \left\lceil \frac{W-t}{\delta}\right\rceil$\;
		}

 	\Return{$v$}\;
}
\caption{Simulating a directed edge preferential attachment network}
\end{algorithm}

We describe an efficient simulation procedure for the preferential attachment network given the parameter values $(\alpha,\beta,\gamma,\deltain,\deltaout)$, where $\alpha+\beta+\gamma=1$. The simulation cost of the algorithm is linear in time. 
{This algorithm, which was provided by Joyjit Roy during his graduate work at Cornell University, is presented below for completeness.}  {Note that this simulation algorithm is specifically designed for the case where the preferential attachment probabilities \eqref{eq:probIn}--\eqref{eq:probOut} are linear in the degrees. A similar idea for the simulation of the Yule-Simon process appeared in \cite{tonelli:concas:locci:2010}. Efficient simulation methods for the case where the preferential attachment probabilities are non-linear are studied in \cite{atwood:2015}, where their algorithm trades some efficiency for the flexibility to model non-linear preferential attachment.}

Using the notation from the introduction, at time $t=0$, we initiate with an arbitrary graph $G(n_0) = (V(n_0),E(n_0))$ of $n_0$ edges, where the elements of $E(n_0)$ are represented in form of $(v_i^{(1)},v_i^{(2)}) \in V(n_0)\times V(n_0)$, $i=1,\ldots,n_0$, with $v_i^{(1)},v_i^{(2)}$ denoting the outgoing and incoming vertices of the edge, respectively. To grow the network, we update the network at each stage from $G(n-1)$ to $G(n)$ by adding a new edge $(v_{n}^{(1)},v_{n}^{(2)})$. Assume that the nodes are labeled using positive integers starting from 1 according to the time order in which they are created, and let the random number $N(n) = |V(n)|$ denote the total number of nodes in $G(n)$.

Let us consider the situation where an existing node is to be chosen from $V(n)$ as the vertex of the new edge. Naively sampling from the multinomial distribution requires $O(N(n))$ evaluations, where $N(n)$ increases linearly with $n$. Therefore the total cost to simulate a network of $n$ edges is $O(n^2)$. This is significantly burdensome when $n$ is large, which is usually the case for observed networks. Algorithm~\ref{algo:1} describes a simulation algorithm which uses the alias method \cite{kronmal:peterson:1979} for node sampling. Here sampling an existing node from $V(n)$ requires only constant execution time, regardless of $n$. Hence the cost to simulate $G(n)$ is only $O(n)$. This method allows generation of a graph with $10^7$ nodes on a personal laptop in less than 5 seconds.

To see that the algorithm indeed produces the intended network, it suffices to consider the case of sampling an existing node from $V(n-1)$ as the incoming vertex of the new edge. In the function \textsf{Node\_Sample} in Algorithm~\ref{algo:1}, we generate $W\sim \text{Uniform}(0,n-1+N(n-1)\deltain)$ and set
$$
v \gets 
v_{\lceil W\rceil}^{(j)}\,\mathbf1_{\{W\le n-1\}} + \left\lceil \frac{W-(n-1)}{\deltain}\right\rceil \,\mathbf1_{\{W>n-1\}}.
$$
Then
\begin{eqnarray*}
\PP\left( v = w \right) &=& \PP\left( v_{\lceil W\rceil}^{(j)} = w \right) \PP\left( W \le n-1 \right)  +  \PP\left( \left\lceil \frac{W-(n-1)}{\deltain}\right\rceil  = w \right) \PP\left( W > n-1 \right) \\
&=& \frac{D_{\text {in}}^{(n-1)}(w)}{n-1} \, \frac{n-1}{n-1 + N(n-1)\deltain} + \frac{1}{N(n-1)} \, \frac{N(n-1) \deltain}{n-1 + N(n-1) \deltain}  \\
&=&  \frac{D_{\text {in}}^{(n-1)}(w) +\deltain}{n-1 + N(n-1) \deltain},
\end{eqnarray*}
which corresponds to the desired selection probability \eqref{eq:probIn}.


\section{Parameter estimation: MLE based on the full network history}\label{sec:estMLE}

In this section, we  estimate  the preferential attachment 
parameter vector $(\alpha,\beta,\din,\dout)$ under two
{assumptions about what data is available.}
In the first scenario, the full evolution of the network is
observed, from which the likelihood function can be computed. 
The resulting MLE is {strongly} consistent and asymptotically {normal}.  For the
second {scenario}, the data only consist of one snapshot of the network
with $n$ edges, without the knowledge of the network history that
produced these edges. {For this scenario we give}
an estimation approach {through approximating the score function and moment matching}, which produces parameter estimators that are
also {strongly} consistent but less efficient than those based on the full
evolution of the network.  
In both cases, {the estimators are uniquely determined.} 


\subsection{Likelihood calculation}\label{MLEdef}
Assume  the network begins with the graph $G(n_0)$ (consisting of
$n_0$ edges) and then evolves {according to the description in
Section \ref{subsec:linpref}} 
with parameters
$(\alpha, \beta, \din, \dout)$, where $\din,\dout>0$ and
$\alpha,\beta$ are non-negative probabilities. The $\gamma$ is
implicitly defined by $\gamma=1-\alpha-\beta$.  
{To avoid trivial cases, we will also assume $\alpha,\beta,\gamma<1$ for the rest of the paper.}
For MLE estimation we restrict the parameter space for $\din,\dout$ to be $[\epsilon,K]$, for some sufficiently small $\epsilon>0$ and large $K$. In particular, the true value of $\din,\dout$ is assumed to be contained in $(\epsilon,K)$.
Let $e_t=(\vstart_t, \vend_t)$ be the newly created edge when the random graph evolves from
$G(t-1)$ to $G(t)$. {W}e sometimes refer to $t$ as the
time rather than the number of edges. 

{Assume} we {observe the initial graph $G(n_0)$ and} the edges $\{e_t\}_{t=n_0+1}^n$ in the order of their formation. For $t=n_0+1,\ldots,n$, the values of the following variables are known:
\begin{itemize}
\item
	$N(t)$, the number of nodes in graph $G(t)$;
\item
	$\Din^{(t-1)}(v)$, $\Dout^{(t-1)}(v)$, the in- and out-degree of node $v$ in $G(t-1)$, for {all} $v\in V(t-1)$;
\item
	$J_t$, the scenario under which $e_t$ is created.
\end{itemize}
Then the  likelihood function is
\begin{align}
L & (\alpha,\beta,\din,\dout|\ G(n_0), (e_t)_{t=n_0+1}^n) \nonumber\\
=& \prod_{t=n_0+1}^n \left(\alpha\frac{\Din^{(t-1)}(\vend_t)+\din}{t-1+\din N(t-1)}\right)^{\ind_{\{J_t=1\}}} \nonumber\\
&\times\prod_{t=n_0+1}^n\left(\beta
  \Bigl(\frac{\Din^{(t-1)}(\vend_t)+\din}{t-1+\din N(t-1)} \Bigr)\Bigl(\frac{\Dout^{(t-1)}(\vstart_t)+\dout}{t-1+\dout N(t-1)}\Bigr)\right)^{\ind_{\{J_t=2\}}}\nonumber\\
&\times\prod_{t=n_0+1}^n \left((1-\alpha-\beta)\frac{\Dout^{(t-1)}(\vstart_t)+\dout}{t-1+\dout N(t-1)}\right)^{\ind_{\{J_t=3\}}}\label{jointL} 
\end{align}
and the log likelihood function is
\begin{align} 
\log&   L(\alpha,\beta,\din,\dout|\ G(n_0),   (e_t)_{t=n_0+1}^n) \label{logL}\\
=& \log\alpha \sum_{t=n_0+1}^n \ind_{\{J_t=1\}} 
+ \log\beta \sum_{t=n_0+1}^n \ind_{\{J_t=2\}} + \log(1-\alpha-\beta) \sum_{t=n_0+1}^n \ind_{\{J_t=3\}}\nonumber\\
&+ \sum_{t=n_0+1}^n \log\left(\Din^{(t-1)}(\vend_t)+\din\right)\ind_{\{J_t\in\lbrace 1,2\rbrace\}} 
+ \sum_{t=n_0+1}^n \log\left(\Dout^{(t-1)}(\vstart_t)+\dout\right)\ind_{\{J_t\in\lbrace 2,3\rbrace\}} \nonumber\\
&- \sum_{t=n_0+1}^n \log(t-1+\din N(t-1))\ind_{\{J_t\in\lbrace 1,2\rbrace\}} 
- \sum_{t=n_0+1}^n \log(t-1+\dout N(t-1))\ind_{\{J_t\in\lbrace 2,3\rbrace\}}.\nonumber
\end{align}
The score functions for $\alpha, \beta, \din, \dout$ are calculated as follows:
\begin{align}
&\frac{\partial}{\partial\alpha} \log L(\alpha,\beta,\din,\dout|\ G(n_0),   (e_t)_{t=n_0+1}^n)
= \frac{1}{\alpha} \sum_{t=n_0+1}^n \ind_{\{J_t=1\}} - \frac{1}{1-\alpha-\beta} \sum_{t=n_0+1}^n \ind_{\{J_t=3\}},
\label{score_alpha}\\
&\frac{\partial}{\partial\beta} \log L(\alpha,\beta,\din,\dout| \ G(n_0),  (e_t)_{t=n_0+1}^n)
= \frac{1}{\beta} \sum_{t=n_0+1}^n \ind_{\{J_t=2\}} - \frac{1}{1-\alpha-\beta} \sum_{t=n_0+1}^n \ind_{\{J_t=3\}},
\label{score_beta}\\
&\frac{\partial}{\partial\din} \log L(\alpha,\beta,\din,\dout|\
  G(n_0),   (e_t)_{t=n_0+1}^n) \label{score_din}\\
&\quad = \sum_{t=n_0+1}^n \frac{1}{\Din^{(t-1)}(\vend_t)+\din}
  \ind_{\{J_t\in\lbrace 1, 2\rbrace\}} 
- \sum_{t=n_0+1}^n \frac{N(t-1)}{t-1+\din N(t-1)}\ind_{\{J_t\in\lbrace 1, 2\rbrace\}} ,
\nonumber \\
&\frac{\partial}{\partial\dout} \log L(\alpha,\beta,\din,\dout|\
  G(n_0),   (e_t)_{t=n_0+1}^n)
  \nonumber\\
&\quad = \sum_{t=n_0+1}^n \frac{1}{\Dout^{(t-1)}(\vstart_t)+\dout}
  \ind_{\{J_t\in\lbrace 2,3\rbrace\}} 
- \sum_{t=n_0+1}^n \frac{N(t-1)}{t-1+\dout N(t-1)}\ind_{\{J_t\in\lbrace 2,3\rbrace\}}. 
\nonumber
\end{align} 

Note that the {score functions \eqref{score_alpha}, \eqref{score_beta} for} $\alpha$ and $\beta$ do not depend on $\din$ and $\dout$. 
One can show that the Hessian matrix of the log-likelihood for $(\alpha,\beta)$ is positive definite. Setting \eqref{score_alpha} and \eqref{score_beta} to zero gives the unique MLE estimates for $\alpha$ and $\beta$,
\begin{align}
\hat{\alpha}^{MLE} &= \frac{1}{n-n_0} \sum_{t=n_0+1}^n \ind_{\{J_t=1\}}\label{alpha_MLE},\\
\hat{\beta}^{MLE} &= \frac{1}{n-n_0} \sum_{t=n_0+1}^n \ind_{\{J_t=2\}}\label{beta_MLE}.
\end{align}
These estimates are strongly consistent by applying the strong law of large numbers for the $\{J_t\}$ sequence.

{Next, consider the first term of the score function for $\din$ in \eqref{score_din}, and we have}
\begin{align*}
 \sum_{t=n_0+1}^n \frac{1}{\Din^{(t-1)}(\vend_t)+\din}\ind_{\left\{J_t\in\lbrace 1, 2\rbrace\right\}}
 &= \sum_{i=0}^\infty \frac{1}{i+\din}\sum_{t=n_0+1}^n \ind_{\left\{\Din^{(t-1)}(\vend_t)=i, J_t\in\lbrace 1, 2\rbrace\right\}}.
\end{align*} 
Observe that $\left\{\Din^{(t-1)}(\vend_t)=i, J_t\in\lbrace 1,
  2\rbrace\right\}$ describes the event that the in-degree of node
$\vend_t {\in V(t-1)}$ is $i$ at time $t-1$ and is augmented to $i+1$ at time
$t$. For {each $i\ge1$, such {an} event happens at some stage $t\in\{n_0+1, n_0+2,\ldots, n\}$ only for those nodes 
with
in-degree $\le i$ at time $n_0$ and in-degree $> i$ at time $n$. }
Let $N_{ij}(n)$ denote the number of nodes with in-degree $i$ and
out-degree $j$ at time $n$, and $\Nin_i(n)$ and $\Nin_{>i}(n)$ to be
the number of nodes with in-degree equal to $i$ and greater than $i$,
respectively, i.e.,
\begin{align*}
 \Nin_i(n) = \sum_{j=0}^\infty N_{ij}(n), & \quad \Nin_{>i}(n) = \sum_{k>i}\Nin_{k} (n).
\end{align*}
Then 
$$\sum_{t=n_0+1}^n \ind_{\left\{\Din^{(t-1)}(\vend_t)=i,J_t\in\lbrace 1, 2\rbrace\right\}}=\Nin_{>i}(n) - \Nin_{>i}(n_0),\quad i\ge 1.$$
On the other hand, when $i=0$,  $\left\{\Din^{(t-1)}(\vend_t)=0,J_t\in\lbrace 1, 2\rbrace\right\}$ occurs for some $t$ if and only if all of the following three events happen: 
\begin{itemize}
\item[(i)]  $\vend_t$ has in-degree $>0$ at time $n$;
\item[(ii)] $\vend_t$ does not have in-degree $>0$ at time $n_0$;
\item[(iii)] $\vend_t$ was not created under the $\gamma$-scheme (otherwise it would have been born with in-degree 1).
\end{itemize}
This implies:
\[
\sum_{t=n_0+1}^n \ind_{\left\{\Din^{(t-1)}(\vend_t)=0,J_t\in\lbrace 1, 2\rbrace\right\}}=\Nin_{>0}(n) - \Nin_{>0}(n_0) - \sum_{t=n_0+1}^n \ind_{\{J_t=3\}},
\]
since there are, in total, $\sum_{t=n_0+1}^n \ind_{\{J_t=3\}}$ nodes
created under the $\gamma$-scheme. {Therefore,}
\begin{align}\label{approx-p1}
\sum_{t=n_0+1}^n \frac{1}{\Din^{(t-1)}(\vend_t)+\din}\ind_{\left\{J_t\in\lbrace 1, 2\rbrace\right\}}
&= \sum_{i=0}^\infty \frac{1}{i+\din}\sum_{t=n_0+1}^n \ind_{\left\{\Din^{(t-1)}(\vend_t)=i, J_t\in\lbrace 1, 2\rbrace\right\}}\nonumber\\
&= \sum_{i=0}^\infty \frac{\Nin_{>i}(n) - \Nin_{>i}(n_0)}{i+\din}-\frac{\sum_{t=n_0+1}^n \ind_{\{J_t=3\}}}{\din}.
\end{align}
Setting {the score function \eqref{score_din}  for $\din$} to 0 and dividing both sides by $n-n_0$ leads to
\begin{align} \label{MLEdin}
\frac{1}{n-n_0}\sum_{i=0}^\infty &\frac{\Nin_{>i}(n) - \Nin_{>i}(n_0)}{i+\din}\nonumber\\
&-\frac{1}{\din(n-n_0)}\sum_{t=n_0+1}^n \ind_{\{J_t=3\}}
-\frac{1}{n-n_0}\sum_{t=n_0+1}^n \frac{N(t-1)}{t-1+\din N(t-1)}\ind_{\{J_t\in\lbrace 1, 2\rbrace\}}=0,
\end{align}
{where the only unknown parameter is $\din$.} 
In Section~\ref{subsec:consistency}, we show that the solution to \eqref{MLEdin} actually maximizes the likelihood function in $\din$.
Similarly, the MLE for $\dout$ can be solved from
\begin{align*} 
\frac{1}{n-n_0}\sum_{j=0}^\infty &\frac{\Nout_{>j}(n) - \Nout_{>j}(n_0)}{j+\dout}\nonumber\\
&-\frac{\frac{1}{n-n_0}\sum_{t=n_0+1}^n \ind_{\{J_t=1\}}}{\dout}
-\frac{1}{n-n_0}\sum_{t=n_0+1}^n \frac{N(t-1)}{t-1+\dout N(t-1)}\ind_{\{J_t\in\lbrace 2, 3\rbrace\}}=0,
\end{align*}
where $\Nout_{>j}(n)$ is defined in the same fashion as $\Nin_{>i}(n)$.

\begin{Remark}\label{sufficiency}
The arguments leading to \eqref{approx-p1} allow us to rewrite the likelihood function \eqref{jointL}:
\begin{align*}
L&(\alpha,\beta,\din,\dout| \ G(n_0), (e_t)_{t=n_0+1}^n) \\
=&\  \alpha^{\sum_{t=n_0+1}^n \ind_{\{J_t=1\}}} 
\ \beta ^{\sum_{t=n_0+1}^n \ind_{\{J_t=2\}}}
\ (1-\alpha-\beta) ^{\sum_{t=n_0+1}^n \ind_{\{J_t=3\}}}\nonumber\\
& \times\prod_{t=n_0+1}^n (t-1+\din N(t-1))^{-\ind_{\left\{J_t\in\lbrace 1, 2\rbrace\right\}} }
\ (t-1+\dout N(t-1)) ^{-\ind_{\left\{J_t\in\lbrace 2,3\rbrace\right\}}}\\
&{\times  \prod_{t=n_0+1}^n \left[\prod_{i=0}^\infty \ (i+\din)^{\ind_{\left\{\Din^{(t-1)}(\vend_t)=i,J_t\in\lbrace 1, 2\rbrace\right\}}} 
\prod_{j=0}^\infty (j+\dout)^{\ind_{\left\{\Dout^{(t-1)}(\vstart_t)=j,J_t\in\lbrace 2, 3\rbrace\right\}}}\right]}\\
= & \alpha^{\sum_{t=n_0+1}^n \ind_{\{J_t=1\}}} 
\ \beta ^{\sum_{t=n_0+1}^n \ind_{\{J_t=2\}}}
\ (1-\alpha-\beta) ^{\sum_{t=n_0+1}^n \ind_{\{J_t=3\}}}\nonumber\\
& \times\prod_{t=n_0+1}^n (t-1+\din N(t-1))^{-\ind_{\left\{J_t\in\lbrace 1, 2\rbrace\right\}} }
\ (t-1+\dout N(t-1)) ^{-\ind_{\left\{J_t\in\lbrace 2,3\rbrace\right\}}}
\ \din^{-\ind_{\{J_t=3\}}}\ \dout^{-\ind_{\{J_t=1\}}}\\
&\times  \prod_{i=0}^\infty \ (i+\din)^{\Nin_{>i}(n)-\Nin_{>i}(n_0)} 
\ {\prod_{j=0}^\infty} (j+\dout)^{\Nout_{>j}(n)-\Nout_{>j}(n_0)}.
\end{align*}
Hence by the factorization theorem, {$N(n_0)$,} $(J_t)_{t=n_0+1}^n$, $(\Nin_{>i}(n)-\Nin_{>i}(n_0))_{i\ge 0}$,  $(\Nout_{>j}(n)-\Nout_{>j}(n_0))_{j\ge 0}$ are sufficient statistics for $(\alpha,\beta,\din, \dout)$.
\end{Remark}


\subsection{Consistency of MLE}\label{subsec:consistency}

{We remarked after}   \eqref{alpha_MLE} and \eqref{beta_MLE} that
  $\hat\alpha^{MLE}$ and $\hat\beta^{MLE}$ converge almost surely to
  $\alpha$ and $\beta$. We now prove that the MLE of $(\din,\dout)$
is also {strongly} consistent.  
Note that if we initiate the network with $G(n_0)$ (for both $n_0$ and $N(n_0)$ finite), then almost surely for all $i,j\ge0$,
\[
\frac{\Nin_{>i}(n_0)}{n}\le\frac{N(n_0)}{n}\to 0, \quad
\frac{\Nout_{>j}(n_0)}{n}\le\frac{N(n_0)}{n}\to 0, 
\quad \text{as }n\to\infty,
\]
and $(n-n_0)/n\to 1$.
In other words, $n_0$, $\Nin_{>i}(n_0)$, $\Nout_{>j}(n_0)$ are all
{$o(n)$.}
{So for simplicity,}
we assume that the graph is initiated with finitely many nodes and no
edge{s}, that is, $n_0=0$ and $N(0)\ge1$. 
In particular, these assumptions imply the sum of the in-degrees at time $n$ is equal to $n$.

Let $\Psi_n(\cdot),\Phi_n(\cdot)$ be the functional forms of the terms in the log-likelihood function \eqref{logL} involving $\din$ and $\dout$ respectively, normalized by $1/n$, i.e.,
\begin{align*}
\Psi_n(\lambda) &:= \sum_{i=0}^\infty \frac{\Nin_{>i}(n)}{n}\log(i+\lambda) - \frac{\log\lambda}{n}\sum_{t=1}^n \ind_{\lbrace J_t=3\rbrace}
 - \frac{1}{n}\sum_{t=1}^n \log\left(t-1+\lambda N(t-1)\right)\ind_{\lbrace J_t\in\{1, 2\}\rbrace},\\
 \Phi_n(\mu) &:= \sum_{j=0}^\infty \frac{\Nout_{>j}(n)}{n}\log(j+\mu) - \frac{\log\mu}{n}\sum_{t=1}^n \ind_{\lbrace J_t=1\rbrace}
 - \frac{1}{n}\sum_{t=1}^n \log\left(t-1+\mu N(t-1)\right)\ind_{\lbrace J_t\in\{2, 3\}\rbrace}.
\end{align*}
The following theorem gives the consistency of the MLE of $\din$ and $\dout$.
\begin{Theorem} \label{thm:consistency}
Suppose $\din, \dout \in (\epsilon,K) \subset (0,\infty)$. 
Define 
\[
\hatdin^{MLE}=\hatdin^{MLE}(n) := \argmax_{\epsilon\le\lambda\le K} \Psi_n(\lambda),\qquad
\hatdout^{MLE}=\hatdout^{MLE}(n) := \argmax_{\epsilon\le\mu\le K} \Phi_n(\mu).
\]
{Then these are the MLE estimators of  $\din, \dout$ and they} are
strongly consistent; that is,
$$
\hatdin^{MLE}\convas \din,\qquad \hatdout^{MLE}\convas \dout, \qquad n \to \infty.
$$
\end{Theorem}

\begin{proof}[Proof of Theorem \ref{thm:consistency}]

{
We only verify the consistency of $\hatdin^{MLE}$ since similar
arguments apply to $\hatdout^{MLE}$. 
Define
\begin{equation}
\psi_n(\lambda)  := \Psi'_n(\lambda) = \sum_{i=0}^\infty \frac{\Nin_{>i}(n)/n}{i+\lambda}-\frac{\frac{1}{n}\sum_{t=1}^n \ind_{\{J_t=3\}}}{\lambda}
-\frac{1}{n}\sum_{t=1}^n \frac{N(t-1)}{t-1+\lambda N(t-1)}\ind_{\{J_t\in\lbrace 1, 2\rbrace\}}.\label{psin}
\end{equation}
Let us consider a limit version of $\psi_n$:
\beqq
\psi(\lambda) :=  \sum_{i=0}^\infty\frac{\pin_{>i}(\din)}{i+\lambda} -\frac{\gamma}{\lambda}
-(1-\beta)a_1(\lambda), \label{defpsi}
\eeqq
where $\pin_{>i}(\din):= \sum_{k>i} \pin_k(\din)$ with $\pin_k(\din):=\pin_k$ as defined in \eqref{asyI}, and
\[
a_1(\lambda) := \frac{\alpha+\beta}{1+\lambda(1-\beta)},\qquad \lambda>0.
\]
Here we write $\pin_{i}(\din)$ to emphasize the dependence on $\din$. 
In Lemmas~\ref{phi} and \ref{unifconv}, provided in the appendix, it is shown that}
{$\psi(\cdot)$ has a unique zero at $\din$,}
{where $\psi(\lambda)>0$ when $\lambda<\din$ and $\psi(\lambda)<0$ when $\lambda>\din$,
and
\beqq\label{eq:unifconv}
	\sup_{\lambda{\ge\epsilon}}|\psi_n(\lambda)-\psi(\lambda)|{\to} 0.
\eeqq
Since $\psi$ is continuous, for any
$\kappa>0$ arbitrarily small, there exists $\varepsilon_\kappa>0$ such
that $\psi(\lambda)>\varepsilon_\kappa$ for $\lambda \in
[\epsilon,\din-\kappa ]$ and $\psi(\lambda)<-\varepsilon_\kappa$ for $\lambda \in [\din+\kappa,K]$. From \eqref{eq:unifconv}, 
\beqq \label{kpbound}
\PP\left(\exists N_\kappa\ s.t. \sup_{n>N_\kappa}\sup_{\lambda\in[\epsilon,K]}|\psi_n(\lambda) - \psi(\lambda)|< \varepsilon_\kappa/2\right) = 1.
\eeqq
Note
 $\sup_{\lambda\in[\epsilon,K]}|\psi_n(\lambda) - \psi(\lambda)|< \varepsilon_\kappa/2$ implies
$$
\psi_n(\lambda) \ge \psi{(\lambda)} - \sup_{\lambda\in[\epsilon,K]}|\psi_n(\lambda) - \psi(\lambda)| \ge \varepsilon_\kappa - \varepsilon_\kappa /2 > 0,\quad \lambda \in [\epsilon,\din-\kappa),
$$
and
$$
\psi_n(\lambda) \le \psi{(\lambda)} + \sup_{\lambda\in[\epsilon,K]}|\psi_n(\lambda) - \psi(\lambda)| \le -\varepsilon_\kappa + \varepsilon_\kappa /2  < 0,\quad \lambda \in (\din+\kappa,K].
$$
These jointly indicate that $\din-\kappa \le \hatdin^{MLE}\le \din+\kappa$.
Hence \eqref{kpbound} implies
$$
\PP\left(\lim_{n\to\infty} |\hatdin^{MLE}-\din|\le \kappa\right) = 1,
$$
for arbitrary $\kappa>0$. That is, $\hatdin^{MLE} \convas \din$.
}
\end{proof}


\subsection{Asymptotic normality of MLE}\label{subsec:asynorm}

In the following theorem, we establish the asymptotic normality for the MLE estimator
$$ \label{mleest}
\hat{\boldsymbol{\theta}}_n^{MLE} = (\hat{\alpha}^{MLE},\, \hat{\beta}^{MLE},\, \hatdin^{MLE},\, \hatdout^{MLE}).
$$

\begin{Theorem}\label{asymp_normality}
Let $\hat{\boldsymbol{\theta}}_n^{MLE} $ be the MLE estimator for $\boldsymbol{\theta}$, the parameter vector of the preferential attachment model. Then
\beqq \label{MLElimit}
\sqrt{n}(\hat{\boldsymbol{\theta}}^\text{MLE}_n - {\boldsymbol{\theta}})\convd N\left(\mathbf0,\Sigma(\boldsymbol{\theta})\right),
\eeqq
where
\beqq \label{fisher}
\Sigma^{-1}(\boldsymbol{\theta}) = I(\boldsymbol{\theta})
:= \begin{bmatrix}
\frac{1-\beta}{\alpha(1-\alpha-\beta)} & \frac{1}{1-\alpha-\beta} & 0 & 0 \\
\frac{1}{1-\alpha-\beta} & \frac{1-\alpha}{\beta(1-\alpha-\beta)} & 0 & 0 \\
0 & 0 & I_\text{in} & 0\\
0 & 0 & 0 & I_\text{out}
\end{bmatrix},
\eeqq
with
\begin{align}
I_\text{in} &:= \sum_{i=0}^\infty \frac{\pin_{>i}}{(i+\din)^2} - \frac{\gamma}{\din^2} - \frac{(\alpha+\beta)(1-\beta)^2}{\left(1+\din(1-\beta)\right)^2}, \label{I_in}\\
I_\text{out} &:= \sum_{j=0}^\infty \frac{\pout_{>j}}{(j+\dout)^2} - \frac{\alpha}{\dout^2} - \frac{(\gamma+\beta)(1-\beta)^2}{\left(1+\dout(1-\beta)\right)^2}. \nonumber
\end{align}
In particular, $I(\boldsymbol{\theta})$ is the asymptotic Fisher information matrix for the parameters, and hence the MLE estimator is efficient.
\end{Theorem}

\begin{Remark}
From Theorem~\ref{asymp_normality}, the estimators
$(\hat{\alpha}^{MLE},\, \hat{\beta}^{MLE})$, $\hatdin^{MLE}$, and
$\hatdout^{MLE}$ are asymptotically independent. 
\end{Remark}

{
\begin{proof} [Proof of Theorem~\ref{asymp_normality}]
We first show the limiting distributions for $(\hat{\alpha}^{MLE},\, \hat{\beta}^{MLE})$, $\hatdin^{MLE}$, and $\hatdout^{MLE}$, respectively.
From \eqref{alpha_MLE} and \eqref{beta_MLE},
$$
(\hat{\alpha}^{MLE},\, \hat{\beta}^{MLE}) = \frac{1}{n} \sum_{t=1}^n \left(\ind_{\{J_t=1\}},\ind_{\{J_t=2\}}\right),
$$
where $\{J_t\}$ is a sequence of iid random variables. Hence the limiting distribution of the pair $\left(\hat{\alpha}^{MLE},\hat{\beta}^{MLE}\right)$ follows directly from standard central limit theorem {for sums of independent random variables}.

Next we show the asymptotic normality for $\hatdin^{MLE}$; the argument for $\hatdout^{MLE}$ is similar.
Recall from \eqref{score_din} that the score function for $\din$ can be written as
$$
\left.\frac{\partial}{\partial\din} \log L(\alpha,\beta,\din,\dout) \right|_\delta =: \sum_{t=1}^n u_t(\delta),
$$
where $u_t$ is defined by
\beqq \label{ut_def}
u_t(\delta):=\frac{1}{\Din^{(t-1)}(\vend_t)+\delta} \ind_{\{J_t\in\lbrace 1, 2\rbrace\}}
- \frac{N(t-1)}{t-1+\delta N(t-1)}\ind_{\{J_t\in\lbrace 1, 2\rbrace\}}.
\eeqq
The MLE estimator $\hatdin^{MLE}$ can be obtained by solving $\sum_{t=1}^n u_t(\delta)=0$. By a Taylor expansion of $\sum_{t=1}^n u_t(\delta)$,
\begin{align}
0 & = \sum_{t=1}^n u_t(\hatdin^{MLE})= \sum_{t=1}^n u_t(\din) + (\hatdin^{MLE} - \din)  \sum_{t=1}^n \dot{u}_t(\hatdin^*), \label{score_taylor}
\end{align}
where $\dot{u}_t$ denotes the derivative of $u_t$ and $\hatdin^*=\din + \xi(\hatdin^{MLE} - \din)$ for some $\xi \in [0,1]$. 
An elementary transformation of \eqref{score_taylor} gives
$$
n^{1/2} (\hatdin^{MLE} - \din) =
\left(- \frac{1}{n^{-1}\sum_{t=1}^n \dot{u}_t(\hatdin^*)} \right) 
\left(n^{-1/2}  \sum_{t=1}^n u_t(\din)  \right).
$$
To establish
$$ 
n^{1/2} (\hatdin^{MLE} - \din) \convd N(0,I_\text{in}^{-1}),
$$
where $I_\text{in}$ is as defined in \eqref{fisher},
it suffices to show the following two results:
\begin{enumerate}
\item[(i)]
	$n^{-1/2}  \sum_{t=1}^n u_t(\din) \convd N(0,I_\text{in})$,
\item[(ii)]
	$n^{-1}\sum_{t=1}^n \dot{u}_t(\hatdin^*) \convp -I_\text{in}$.
\end{enumerate}
These are proved in Lemmas~\ref{normality_lemma2} and \ref{normality_lemma3} in the appendix, respectively.

To establish the joint asymptotic normality of the MLE estimator $\hat{\boldsymbol{\theta}}_n^{MLE}$, denote the joint score function vector for $\boldsymbol{\theta}$ by
$$
\frac{\partial}{\partial\boldsymbol{\theta}} \log L(\boldsymbol{\theta}) =: \mathbf S_n(\boldsymbol{\theta}) = \left(S_n(\alpha),S_n(\beta),S_n(\din),S_n(\dout)\right)^T,
$$
where $S_n(\alpha),S_n(\beta),S_n(\din),S_n(\dout)$ are the score functions for $\alpha,\beta,\din,\dout$, respectively. A multivariate Taylor expansion gives
\beqq \label{joint_score_taylor}
\mathbf0 = \mathbf S_n\left(\hat{\boldsymbol{\theta}}_n^{MLE}\right) = \mathbf S_n(\boldsymbol{\theta}) + \dot{\mathbf S}_n\left(\hat{\boldsymbol{\theta}}_n^*\right) \left(\hat{\boldsymbol{\theta}}_n^{MLE} -\boldsymbol{\theta} \right),
\eeqq
where $ \dot{\mathbf S}_n$ denotes the Hessian matrix of the log-likelihood function $\log L(\boldsymbol{\theta})$, and {$\hat{\boldsymbol{\theta}}_n^* = \boldsymbol\theta + \boldsymbol\xi \circ\left(\hat{\boldsymbol{\theta}}_n^{MLE} -\boldsymbol{\theta} \right)$ for some vector $\boldsymbol\xi \in [0,1]^4$, where ``$\circ$" denotes the Hadamard product.} From Remark~\ref{sufficiency}, the likelihood function $L(\boldsymbol{\theta})$ can be factored into
$$
L(\boldsymbol{\theta}) = f_1(\alpha,\beta)f_2(\din)f_3(\dout).
$$
Hence
\beqq \label{score_hes_conv}
\frac{1}{n}\dot{\mathbf S}_n(\hat{\boldsymbol{\theta}}_n^*) = 
\begin{bmatrix}
\frac{\partial^2\log L_n(\hat{\boldsymbol{\theta}}_n^*)}{\partial\alpha^2} & \frac{\partial^2\log L_n(\hat{\boldsymbol{\theta}}_n^*)}{\partial\alpha\partial\beta} & 0 & 0 \\
\frac{\partial^2\log L_n(\hat{\boldsymbol{\theta}}_n^*)}{\partial\beta\partial\alpha} & \frac{\partial^2\log L_n(\hat{\boldsymbol{\theta}}_n^*)}{\partial\beta^2} & 0 & 0 \\
0 & 0 &\frac{\partial^2\log L_n(\hat{\boldsymbol{\theta}}_n^*)}{\partial\din^2} & 0\\
0 & 0 & 0 & \frac{\partial^2\log L_n(\hat{\boldsymbol{\theta}}_n^*)}{\partial\dout^2}
\end{bmatrix}
\convp I(\boldsymbol{\theta})
\eeqq
as implied in the previous part of the proof, where $I(\boldsymbol{\theta})$ is as defined in \eqref{fisher} and is positive semi-definite. 

Note that $(S_n(\alpha),S_n(\beta)),S_n(\din),S_n(\dout)$ are pairwise uncorrelated. As an example, observe that
\begin{align*}
\EE[S_n(\alpha)S_n(\din)] =&\ \int \frac{\partial\log L(\boldsymbol{\theta})}{\partial\alpha}\frac{\partial\log L(\boldsymbol{\theta})}{\partial\din} L(\boldsymbol{\theta})d\mathbf{x}\\
 =&\ \int\frac{\partial\log f_1(\alpha,\beta)}{\partial\alpha}\frac{\partial\log f_2(\din)}{\partial\din}  f_1(\alpha,\beta)f_2(\din)f_3(\dout) d\mathbf{x} \\
 =&\ \int\frac{\partial f_1(\alpha,\beta)}{\partial\alpha}\frac{\partial  f_2(\din)}{\partial\din} f_3(\dout) d\mathbf{x} \\
 =&\ \frac{\partial^2}{\partial\alpha\partial\din} \int  L(\boldsymbol{\theta}) d\mathbf{x}  \\
 = &\ 0 
 = \EE[S_n(\alpha)]\EE[S_n(\din)].
\end{align*}
Using the Cram\'er-Wold device, the joint convergence of $\mathbf S_n(\boldsymbol{\theta})$ follows easily, i.e.,
$$
n^{-1/2} \mathbf S_n(\boldsymbol{\theta}) \convd N(\mathbf0, I(\boldsymbol{\theta})).
$$
From here, the result of the theorem follows from \eqref{joint_score_taylor} and \eqref{score_hes_conv}.\end{proof}
}


\section{Parameter estimation based on one snapshot}\label{OneSnapshot}

Based only {on} the single snapshot $G(n)$, we propose
a parameter estimation procedure.
We assume that the choice of the snapshot does not depend on any
  endogenous information related to the network. The snapshot merely
  represents a point in time where the data is available.
Since no information on
the initial graph $G(n_0)$ is {available, we merely assume $n_0$ and $N(n_0)$ are fixed and $n\to\infty$}.

Among the sufficient statistics for $(\alpha,\beta,\din,\dout)$ derived in Remark \ref{sufficiency}, $\left(\Nin_{>i}(n)\right)_{i\ge 0}$,  $\left(\Nout_{>j}(n)\right)_{j\ge 0}$ are computable from $G(n)$, but the $(J_t)_{t=1}^n$ are not. However, when $n$ is large, we can use the following approximations according to the proof of {Lemma}~\ref{unifconv}:
$$
\frac{1}{n}\sum_{t=n_0+1}^n \ind_{\{J_t=3\}} \approx 1-\alpha-\beta,
$$
and
\[
\frac{1}{n}\sum_{t=n_0+1}^n \frac{N(t)}{t+\din N(t)}\ind_{\{J_t\in\lbrace 1, 2\rbrace\} }
\approx (\alpha + \beta) \frac{1-\beta}{1+\din (1-\beta)}.
\]
Substituting in \eqref{MLEdin}, we {estimate $\din$ in terms of $\alpha$ and $\beta$ by solving}
\beqq \label{approx_din}
\sum_{i=0}^\infty \frac{\Nin_{>i}(n)/n}{i+\din}-\frac{1-\alpha-\beta}{\din}
-\frac{(\alpha+\beta)(1-\beta)}{1+(1-\beta)\din} = 0.
\eeqq
Note that a strongly consistent estimator of $\beta$ can be obtained directly from $G(n)$:
$$\label{tildebeta}
\tilde\beta= 1-\frac{N(n)}{n} 
\convas \beta.
$$
To obtain an estimate for $\alpha$, we make use of the recursive formula for $\{\pin_i\}$ in \eqref{rec_pin1}:
\beqq\label{alpha-p0}
\left(1+\frac{(\alpha+\beta)\din}{1+(1-\beta)\din}\right)\pin_0 = \alpha,
\eeqq
and replace $\pin_0$ by $\Nin_0(n)/n$ for large $n$,
\beqq\label{alpha1}
\left(1+\frac{(\alpha+\beta)\din}{1+(1-\beta)\din}\right)\frac{\Nin_0(n)}{n} = \alpha.
\eeqq
Plug the strongly consistent estimator $\tilde{\beta}$ into \eqref{approx_din} and \eqref{alpha1}, and
we claim that solving the system of equations:
\begin{subequations}
\label{onesnap}
\begin{align}
 & \sum_{i=0}^\infty \frac{\Nin_{>i}(n)/n}{i+\din}-\frac{1-\alpha-\tilde{\beta}}{\din}
-\frac{(\alpha+\tilde{\beta})(1-\tilde{\beta})}{1+(1-\tilde{\beta})\din} = 0,\label{onesnap1}\\
& \left(1+\frac{(\alpha+\tilde{\beta})\din}{1+(1-\tilde{\beta})\din}\right)\frac{\Nin_0(n)}{n} = \alpha, \label{onesnap2}
\end{align}
\end{subequations}
gives the unique solution $(\tilde{\alpha},\tildedin)$ which is strongly consistent for $(\alpha,\din)$.

\begin{Theorem} \label{ss_consist}
The solution $(\tilde{\alpha},\tildedin)$ to {the system of equations in \eqref{onesnap}} is unique and strongly consistent
for $(\alpha,\din)$, i.e.
$$ \tilde{\alpha}\convas\alpha,\quad \tildedin\convas\din.$$
\end{Theorem}

The proof of Theorem~\ref{ss_consist} is given in Section~\ref{subsec:proof3}.
\vspace{.1in}

The parameters $\tildedout$ and $\tilde\gamma$ can be estimated by a mirror argument. We summarize the estimation procedure for $(\alpha,\beta,\gamma,\din,\dout)$ from the snapshot $G(n)$ as follows:
\begin{enumerate}
\item[1.] {Estimate $\beta$} by $\tilde{\beta}=1-N(n)/n$.
\item[2.] Obtain $\tildedin^0$ by solving {(i.e., matching \eqref{onesnap1} and \eqref{onesnap2})}
$$
\sum_{i=1}^\infty \frac{\Nin_{>i}(n)}{n}\frac{i}{i+\din}(1+\din(1-\tilde\beta)) =\frac{\frac{\Nin_0(n)}{n} + \tilde\beta }{1-\frac{\Nin_0(n)}{n} \frac{\din}{1+(1-\tilde\beta)\din}}.
$$
\item[3.] {Estimate $\alpha$} by
$$
\tilde\alpha^0 =   \frac{\frac{\Nin_0(n)}{n} + \tilde\beta}{1-\frac{\Nin_0(n)}{n} \frac{\tildedin^0}{1+(1-\tilde\beta)\tildedin^0}} - \tilde\beta.
$$
\item[4.] Obtain $\tildedout^0$ by solving
$$
\sum_{j=1}^\infty \frac{\Nout_{>j}(n)}{n}\frac{j}{j+\dout}(1+\dout(1-\tilde\beta)) = \frac{\frac{\Nout_0(n)}{n} + \tilde\beta }{1-\frac{\Nout_0(n)}{n} \frac{\dout}{1+(1-\tilde\beta)\dout}}.
$$
\item[5.] {Estimate $\gamma$} by
$$
\tilde\gamma^0 =  \frac{\frac{\Nout_0(n)}{n} + \tilde\beta }{1-\frac{\Nout_0(n)}{n} \frac{\tildedout^0}{1+(1-\tilde\beta)\tildedout^0}} - \tilde\beta.
$$
\end{enumerate}
\smallskip
{
Note that even though all three estimators $\tilde\alpha^0,\tilde\beta,\tilde\gamma^0$ are strongly consistent and hence $\tilde\alpha^0+\tilde\beta+\tilde\gamma^0\convas1$, Step 1--5 do not necessarily imply the strict equality
\beqq \label{eq:probeq}
	\tilde\alpha^0+\tilde\beta+\tilde\gamma^0=1.
\eeqq
We recommend adding the following two steps for a re-normalization to overcome this defect.
}
\begin{enumerate}
\item[6.]  Re-normalize the probabilities
$$ \label{eq:renormalize}
	(\tilde\alpha,\tilde\beta,\tilde\gamma) \leftarrow \left(\frac{\tilde\alpha^0(1-\tilde\beta)}{\tilde\alpha^0+\tilde\gamma^0},\tilde\beta,\frac{\tilde\gamma^0(1-\tilde\beta)}{\tilde\alpha^0+\tilde\gamma^0}\right).
$$
\item[7.] Plug $\tilde\alpha$ into \eqref{onesnap1} to update the estimate of $\din$, i.e., solve for $\tildedin$ from
$$
	\sum_{i=0}^\infty \frac{\Nin_{>i}(n)/n}{i+\tildedin}-\frac{1-\tilde\alpha-\tilde{\beta}}{\tildedin}
-\frac{(\tilde\alpha+\tilde{\beta})(1-\tilde{\beta})}{1+(1-\tilde{\beta})\tildedin} = 0.
$$
Similarly, solve for $\tildedout$ from
$$
	\sum_{j=0}^\infty \frac{\Nout_{>j}(n)/n}{j+\tildedout}-\frac{1-\tilde\gamma-\tilde{\beta}}{\tildedout}
-\frac{(\tilde\gamma+\tilde{\beta})(1-\tilde{\beta})}{1+(1-\tilde{\beta})\tildedout} = 0.
$$
\end{enumerate}


\section{Simulation study}\label{sec:estSim}
We now apply the estimation procedures described in Sections~\ref{sec:estMLE} and \ref{OneSnapshot} to simulated data,
which allows us to compare the estimation results using the full history of the network with {that using} just one snapshot.
Algorithm~\ref{algo:1} is used to simulate realizations of the preferential attachment network.

\subsection{MLE}\label{estSim:MLE}
\begin{figure}[t]\center
\includegraphics[width=12cm, height=10cm]{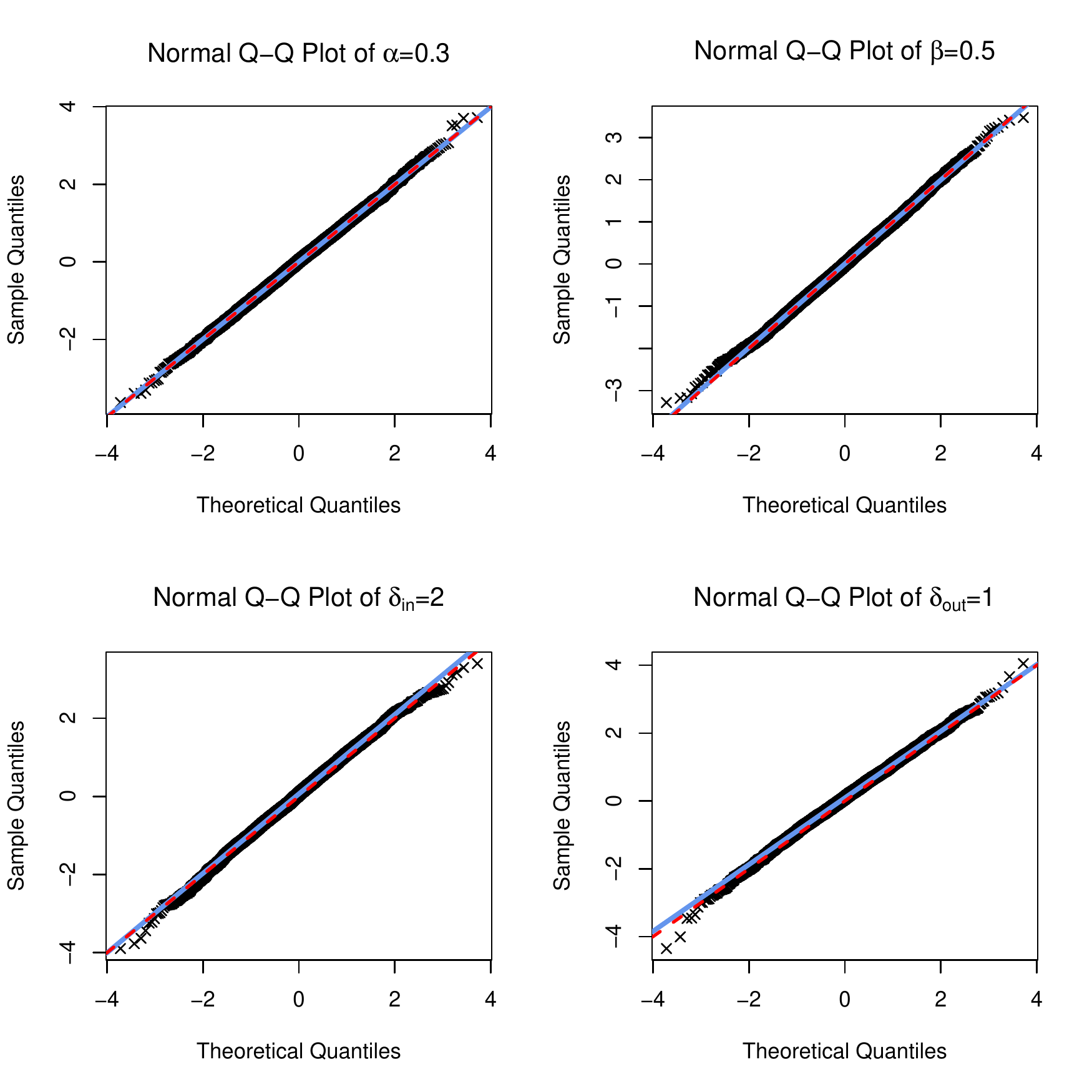}
\caption{Normal QQ-plots {in black} for normalized estimates in \eqref{normalize} under $5000$ replications of a preferential attachment network with $10^5$ edges and $\boldsymbol{\theta}=(0.3,0.5,2,1)$. The fitted lines in {blue} are the traditional qq-lines {(given by R)} used to check normality of the estimates.
The red dashed line represents the $y=x$ line in all plots.}\label{qq-MLE}
\end{figure}

For the scenario of observing the full history of the network, we simulated 5000 independent replications of the preferential attachment network with $10^5$ edges under the true parameter values
\beqq\label{params}
\boldsymbol{\theta} = \left(\alpha,\beta,\din, \dout\right) = (0.3,\, 0.5,\, 2,\, 1).
\eeqq
For each realization, the MLE estimate of the parameters was computed and standardized as
{
\beqq\label{normalize}
\frac{\sqrt{n}\left((\hat{\boldsymbol{\theta}}_n^{MLE})_i-(\boldsymbol{\theta})_i\right)}{\hat{\sigma}_{ii}},
\eeqq
where} $(\hat{\boldsymbol{\theta}}_n)_i$ and $(\boldsymbol{\theta})_i$ denote the $i$-th components of $\hat{\boldsymbol{\theta}}_n^{MLE}$ and $\boldsymbol{\theta}$ respectively, and 
$\hat{\sigma}_{ii}^2$ is the $i$-th diagonal component of the matrix $\hat{\Sigma}:= \Sigma(\hat{\boldsymbol{\theta}}_n^{MLE})$. 
{The explicit formula for the entries of $\hat{\Sigma}$ is}
$$
\hat{\Sigma} = 
\begin{bmatrix}
\hat{\alpha}^{MLE}\left(1-\hat{\alpha}^{MLE}\right) & -\hat{\alpha}^{MLE}\hat{\beta}^{MLE}&0 & 0\\
-\hat{\alpha}^{MLE}\hat{\beta}^{MLE} & \hat{\beta}^{MLE}\left(1-\hat{\beta}^{MLE}\right) & 0 & 0 \\
0 & 0 & \hat{I}^{-1}_\text{in} & 0 \\
0 & 0 & 0 & \hat{I}^{-1}_\text{out}
\end{bmatrix},
$$
where, {see \eqref{fisher} and \eqref{I_in},}
\begin{align*}
\hat{I}_\text{in} &= \sum_{i=0}^\infty \frac{\Nin_{>i}(n)/n}{\left(i+\hatdin^{MLE}\right)^2} - \frac{1-\hat{\alpha}^{MLE}-\hat{\beta}^{MLE}}{\left(\hatdin^{MLE}\right)^2} - \frac{\left(\hat{\alpha}^{MLE}+\hat{\beta}^{MLE}\right)\left(1-\hat{\beta}^{MLE}\right)^2}{\left(1+\hatdin^{MLE}\left(1-\hat{\beta}^{MLE}\right)\right)^2},\\
\hat{I}_\text{out} &= \sum_{j=0}^\infty \frac{\Nout_{>j}(n)/n}{\left(j+\hatdout^{MLE}\right)^2} - \frac{\hat{\alpha}^{MLE}}{\left(\hatdout^{MLE}\right)^2} - \frac{\left(1-\hat{\alpha}^{MLE}\right)\left(1-\hat{\beta}^{MLE}\right)^2}{\left(1+\hatdout^{MLE}(1-\hat{\beta}^{MLE})\right)^2}.
\end{align*}
{By the consistency of the MLEs combined with the convergence of $\{\Nin_{>i}(n)/n\}$ and $\{\Nout_{>j}(n)/n\}$, see \eqref{eq:Nin_over}, we have that $\hat{\Sigma}_n \overset{a.s.}\to \Sigma$.

The QQ-plots of the normalized MLEs are shown in  Figure~\ref{qq-MLE}, all of which line up quite well with the $y=x$ line (the red dashed line).
This is consistent with the asymptotic theory described in Theorem~\ref{asymp_normality}.
Confidence intervals for $\boldsymbol{\theta}$ can be obtained using this theorem.}
Given a single realization, {an approximate $(1-\varepsilon)$-confidence interval} for $(\boldsymbol{\theta})_i$ is
$$
(\hat{\boldsymbol{\theta}}_n^{MLE})_i \pm z_{\varepsilon/2}\sqrt{\frac{\hat{\sigma}_{ii}^2}{n}}\quad\mbox{for } i=1,\ldots,4,
$$
where $z_{\varepsilon/2}$ is the upper $\varepsilon/2$ quantile of $N(0, 1)$.

\subsection{One snapshot}\label{subsec:oness}
\begin{figure}[t]\center
\includegraphics[width=12cm, height =10cm]{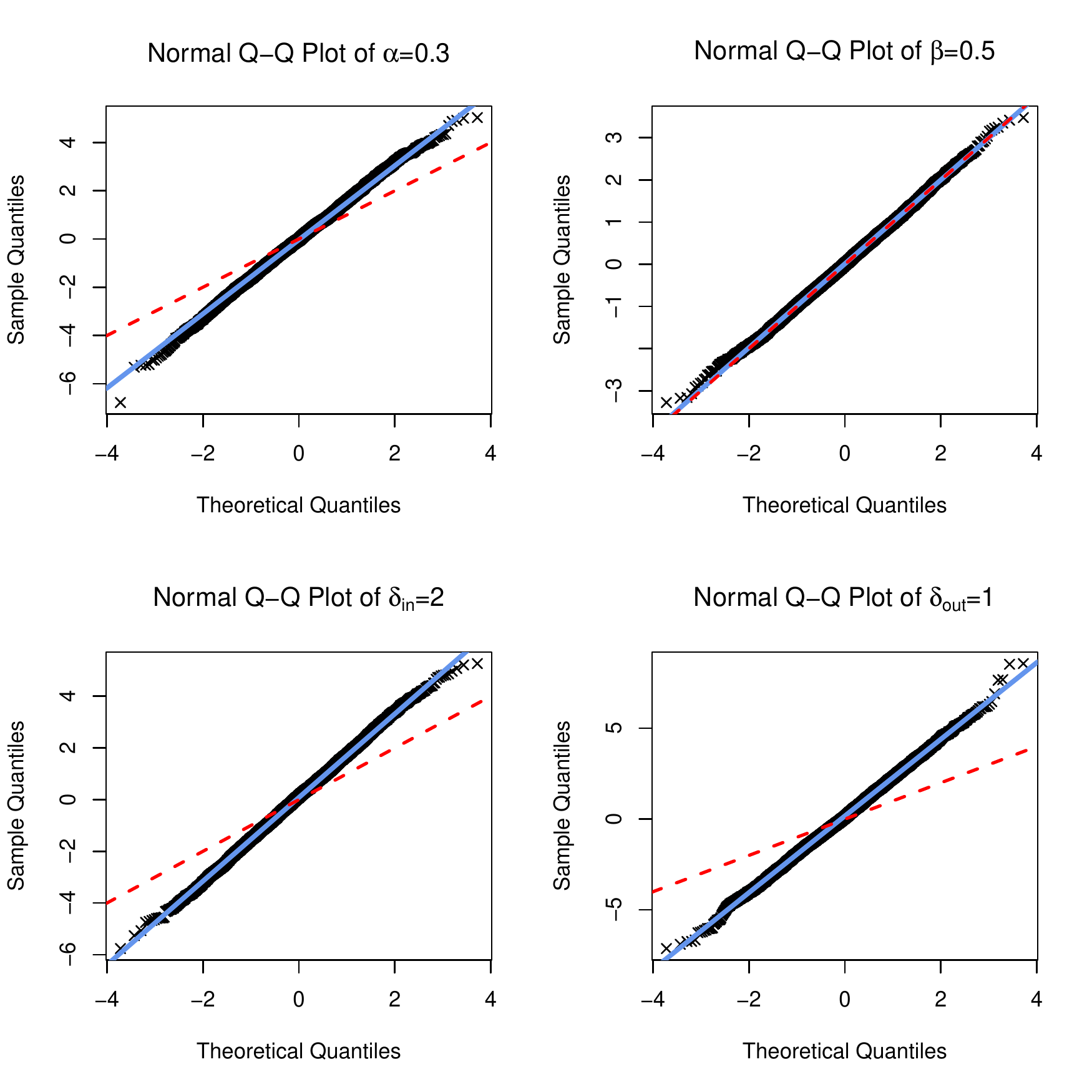}
\caption{Normal QQ-plots for the normalized estimates in \eqref{normalize-oness} under $5000$ replications of a preferential attachment network with $10^5$ edges and $\boldsymbol{\theta}=(0.3,0.5,2,1)$. The fitted lines in {blue} are the traditional qq-lines used to check normality of the estimates.
The red dashed line represents the $y=x$ line in all plots.}\label{qq-onesnapshot}
\end{figure}

We used the same simulated data as in Section~\ref{estSim:MLE} to obtain parameter estimates {$\tilde{\boldsymbol{\theta}}_n := (\tilde{\alpha},\, \tilde{\beta},\, \tildedin,\, \tildedout)$} through only the final snapshot, {i.e.{,} the set of directed edges without timestamps,} following the procedure described at the end of Section~\ref{OneSnapshot}. 
{For the purpose of comparison with MLE, Figure~\ref{qq-onesnapshot} gives the QQ-plots for the normalized estimates from the snapshots using the same standardizations for the MLEs, i.e.,
\beqq\label{normalize-oness}
\frac{\sqrt{n}\left((\tilde{\boldsymbol{\theta}}_n)_i-(\boldsymbol{\theta})_i\right)}{\hat{\sigma}_{ii}},\quad i=1,\ldots,4,
\eeqq
where $(\tilde{\boldsymbol{\theta}}_n)_i$ denotes the $i$-th components of $\tilde{\boldsymbol{\theta}}_n$.}
Again, the fitted lines in {blue} are the traditional QQ-lines and the red dashed lines are the $y=x$ line.
The QQ-plot for $\tilde{\beta}$ exhibits {the} same shape as for $\hat{\beta}^{MLE}$, {since the two estimates are identical.}

{
From Figure~\ref{qq-onesnapshot}, we see that the snapshot estimates of all four parameters are consistent and approximately normal, i.e., the QQ-plots are linear. However, the slope{s} of the QQ-lines for $\tilde{\alpha},\tildedin, \tildedout$ are much steeper than the diagonal line, indicating a loss of efficiency for $\tilde{\boldsymbol{\theta}}_n$ compared with $\hat{\boldsymbol{\theta}}_n$. Indeed the estimator variance is inflated for all parameters  except for $\beta$, where $\tilde\beta$ coincides with the true MLE.  This is as expected since knowing only the final snapshot provides far less information than the whole network history.
}

Recall that for a consistent estimator $T_n$ of a one-dimensional parameter $\theta$ constructed from a random sample of size $n$, the asymptotic relative efficiencies (ARE) of $T_n$ is defined by
\[
ARE(T_n) := \lim_{n\to\infty}\frac{\Var(\sqrt{n}T_n^*)}{\Var(\sqrt{n}T_n)},
\]
where $T_n^*$ denotes the asymptotically efficient estimator.
{We may compute the ARE's for the snapshot parameter estimates }
\begin{subequations}
\label{eff}
\begin{align}
ARE(\tilde{\alpha}) & =\lim_{n\to\infty}\frac{n\Var(\hat{\alpha}^{MLE})}{n\Var(\tilde{\alpha})}  
\approx \frac{\widehat{\Var}(\hat{\alpha}^{MLE})}{\widehat{\Var}(\tilde{\alpha})}  
\approx 0.398, \\ 
ARE(\tildedin)&=\lim_{n\to\infty}\frac{n\Var(\hatdin^{MLE})}{n\Var(\tildedin)} 
\approx \frac{\widehat{\Var}(\hatdin^{MLE})}{\widehat{\Var}(\tildedin)} 
\approx 0.392,\\
ARE(\tildedout) & = \lim_{n\to\infty}\frac{n\Var(\hatdout^{MLE})}{n\Var(\tildedout)} 
\approx \frac{\widehat{\Var}(\hatdout^{MLE})}{\widehat{\Var}(\tildedout)}
\approx 0.226,
\end{align} 
\end{subequations}
{where $\widehat\Var$ denotes the {sample} variance of the parameter estimate based on the 5000 replications. Note that $ARE(\tilde\beta)=1$ since $\tilde\beta = \hat\beta^{MLE}$.}

Given a single realization, the variances of the snapshot estimates can be estimated through resampling as follows.
{
Using the estimated parameter  $\tilde{\boldsymbol{\theta}}_n$,
simulate $10^4$ independent bootstrap replicates of the network with $n=10^5$ edges.
{For each simulated network, the snapshot estimate,
  $\tilde{\boldsymbol{\theta}}^*_n := \left(\tilde{\alpha}^*,\,
    \tilde{\beta}^*,\, \tildedin^*,\, \tildedout^*\right)$, is
  computed.  The sample variance of these $10^4$ snapshot estimates
  can then be used  as an approximation for the variance of
  $\tilde{\boldsymbol{\theta}}_n$ so that assuming asymptotic
  normality, a $(1-\varepsilon)$-confidence interval for $\btheta$ can
  be approximated by 
}
$$
{(\tilde{\boldsymbol{\theta}}_n)_i\pm z_{\varepsilon/2}\sqrt{\widehat{\Var}\left((\tilde{\boldsymbol{\theta}}^*_n)_i\right)}\quad\mbox{for } i=1,\ldots,4,}
$$
where $z_{\varepsilon/2}$ is the upper $\varepsilon/2$ quantile of $N(0, 1)$.}

\subsection{Sensitivity test}
Now {we {investigate} the sensitivity {of} our estimates while values of
  the parameters $(n,\alpha, \beta,\din, \dout)$ are allowed to vary.}
First consider the impact of $n$, the number of edges in the network. {To do so we held the parameters {fixed with values given by} \eqref{params}: $\left(\alpha,\beta,\din, \dout\right) = (0.3,\, 0.5,\, 2,\, 1)$ and varied the value of $n$.} The QQ-plots {(not presented)} for standardized estimates using both full MLE and one-snapshot methods {were} produced to check the asymptotic normality. {When $n = 500,1000$, diagnostics revealed departures from normality for {both} the MLE {and} the snapshot estimates. However, after increasing $n$ to $10000$, estimates obtained from both approaches appeared normally distributed as expected.}

{For each value of $n$ in Table~\ref{varyn}, 5000 replicates of the network with $n$ edges and parameters $\boldsymbol{\theta}=(0.3, 0.5, 2,1)$ were generated.
For each realization, the MLE's $\hat{\boldsymbol{\theta}}_n^{MLE}$ were computed using the full history of the network and the one-snapshot estimates $\tilde{\boldsymbol{\theta}}_n$ were obtained using the 7-step snapshot method proposed in Section~\ref{OneSnapshot}, pretending that only the last snapshot $G(n)$ was available. The mean for these two estimators were recorded in Table~\ref{varyn}.}
{There is little bias for both estimates of $\alpha$ and $\beta$, even for small values of $n$. On the other hand, there is some bias for estimated $\din$ and $\dout$ for $n\le 5000$. The magnitude of the biases for both types of estimates {decrease} as $n$ increases.}
Also the ARE's of the snapshot estimator {stay within a narrow band} as $n$ increases.

\begin{table}[t]
\centering
\caption{\smallskip{Mean of $\hat{\boldsymbol{\theta}}_n^{MLE}$ and $\tilde{\boldsymbol{\theta}}_n$ with ARE's of $\tilde{\boldsymbol{\theta}}_n$ relative to $\hat{\boldsymbol{\theta}}_n^{MLE}$ for $\boldsymbol{\theta}=(0.3, 0.5, 2,1)$ under different choices of $n$.}}\label{varyn}

\begin{tabular}{ccccc}
 \hline
 $n$&  $Mean(\hat{\boldsymbol{\theta}}_n^{MLE})$ & 
 $Mean(\tilde{\boldsymbol{\theta}}_n)$ &  $ARE(\tilde{\boldsymbol{\theta}}_n)$\\
 \hline
  $1000$ & (0.300, 0.500, {2.076}, {1.054}) & ({0.301}, 0.500, {2.128}, {1.066}) & (0.408, 1.000, 0.397, 0.228) \\
  $5000$ & (0.300, 0.500, {2.022}, {1.013}) & ({0.301}, 0.500, {2.036}, {1.010}) & (0.414, 1.000, 0.386, 0.236) \\
  $10000$ & (0.300, 0.500, {2.011}, {1.006}) & (0.301, 0.500, {2.019}, {1.006}) & (0.408, 1.000, 0.388, 0.232) \\
   $50000$ & (0.300, 0.500, {2.003}, {1.002}) & (0.300, 0.500, {2.005}, {1.002}) & (0.399, 1.000, 0.393, 0.230) \\
  $100000$ & (0.300, 0.500, 2.001, 1.001) & (0.300, 0.500, 2.003, 1.000) & (0.392, 1.000, 0.382, 0.223)\\
   \hline
\end{tabular}
\end{table}



{
Next we held $(n, \din, \dout) = (10^5, 2, 1)$ fixed and experimented with various values of $(\alpha, \beta)$ in Table~\ref{varyab}.
For each choice of $(\alpha, \beta)$, 5000 independent realizations of the network were generated and the means of the MLE $\hat{\boldsymbol{\theta}}_n^{MLE}$ and the one-snapshot estimates $\tilde{\boldsymbol{\theta}}_n$ were recorded.}
Overall, the biases for $\hat{\boldsymbol{\theta}}_n^{MLE}$ are remarkably small for virtually all combinations of parameter values, except for those parameter choices where one of $(\alpha, \beta)$ is extremely small. The biases for the snapshot estimates $\tilde{\boldsymbol{\theta}}_n$ exhibit a similar property, but the magnitudes of the biases are consistently larger than those in the MLE case. 

{
In general, the snapshot estimators are able to achieve $20\%$--$50\%$ efficiency over the range of parameters considered.  The loss of efficiency might be less than one would expect given the substantial reduction in the data available to produce the snapshot estimates. It is worth noting that in the case where $(\alpha,\beta)=(0.7,0.2)$, the efficiencies of the snapshot estimators for $\alpha$ and $\din$ are much larger (.73 and .79, respectively).  A heuristic explanation for this increase is that the parameter $\gamma =1-\alpha-\beta=0.1$ is relatively small.  By the implicit constraints used for the snapshot estimates, we have 
$$
	\tilde\alpha+\tilde\gamma=1-\tilde \beta=1-\hat\beta^{MLE}=\hat\alpha^{MLE}+\hat\gamma^{MLE},
$$
that is, the snapshot estimate of the sum $\alpha+\gamma$ is the same as the MLE for the sum.  Now if $\gamma$ is small, one would expect the resulting estimates to also be small so that $\tilde\alpha$ would be nearly the same as $\hat \alpha^{MLE}$.  Hence the ARE would be close to 1.  On the other hand, in the case of a larger $\gamma$, see the bottom row of Table~\ref{varyab} in which $\gamma=0.6$, the ARE for $\alpha$ is not as large (.42), but the ARE for $\tildedout$ is (.63).}

\begin{table}[t]
\caption{\smallskip{Mean of $\hat{\boldsymbol{\theta}}_n^{MLE}$ and $\tilde{\boldsymbol{\theta}}_n$ with ARE's of $\tilde{\boldsymbol{\theta}}_n$ relative to $\hat{\boldsymbol{\theta}}_n^{MLE}$ for $(n,\din,\dout)=(10^5,2,1)$ under different choices of $(\alpha,\beta)$.}}
\label{varyab}
\centering
\begin{tabular}{ccccc}
 \hline
 $(\alpha, \beta)$ & $Mean(\hat{\boldsymbol{\theta}}_n^{MLE})$ & 
 $Mean(\tilde{\boldsymbol{\theta}}_n)$ &  $ARE(\tilde{\boldsymbol{\theta}}_n)$\\
 \hline
  (0.001, 0.99) & (0.001, 0.990, {2.034}, {1.016}) & (0.001, 0.990, {2.071}, {1.049}) & (0.291, 1.000, 0.147, 0.316) \\
  (0.01, 0.9) & (0.010, 0.900, {2.004}, 1.001) & (0.010, 0.900, {2.008}, {1.004}) & (0.331, 1.000, 0.207, 0.381) \\
  (0.1, 0.8) & (0.100, 0.800, 2.003, 1.001) & (0.100, 0.800, {2.004}, {1.002}) & (0.353, 1.000. 0.264, 0.216)\\
  (0.2, 0.6) & (0.200, 0.600, 2.002, 1.001) &  (0.200, 0.600, 2.003, 1.001)	& (0.364, 1.000, 0.309, 0.236)	\\
  (0.5, 0.3) & (0.500, 0.300, 2.001, 1.001) & (0.500, 0.300, {2.002}, 1.000) & (0.472, 1.000, 0.529, 0.202)		\\
  (0.7, 0.2) & (0.700, 0.200, 2.002, 1.000) & ({0.700}, 0.200, {2.002}, 1.000)	& (0.726, 1.000, 0.793, 0.217)	\\
    (0.1, 0.3) & (0.100, 0.300, 2.001, 1.001) & ({0.100}, 0.300, {2.002}, 1.000)	& (0.420, 1.000, 0.313, 0.629)	\\
   \hline
\end{tabular}
\end{table}


\section{Real network example}\label{sec:estReal}

In this section, we explore fitting a preferential attachment model to
a social network. As illustration, we chose the Dutch Wiki talk
network dataset, available on KONECT \cite{kunegis:2013}
(\url{http://konect.uni-koblenz.de/networks/wiki_talk_nl}). The nodes
represent users of Dutch Wikipedia, and an edge from node A to node B
refers to user A writing a message on the talk page of user B at a
certain time point. The network consists of 225,749 nodes (users) and
1,554,699 edges (messages). All edges are recorded with timestamps. 

In order to accommodate all the edge formulation scenarios appeared in the dataset, we extend our model by appending the following two interaction schemes ($J_n=4,5$) in addition to the existing three ($J_n=1,2,3$) described in Section \ref{subsec:linpref}. 
\begin{itemize}
\item If $J_n=4$ (with probability $\xi$), append to $G(n-1)$ two new nodes $v,w\in V(n)\setminus V(n-1)$ and an edge connecting them $(v,w)$.
\item  If $J_n=5$ (with probability $\rho$),  append to $G(n-1)$ a new node $v\in V(n)\setminus V(n-1)$ {with} self loop $(v,v)$.
\end{itemize}
These scenarios have been observed in other social network data, such as the Facebook wall post network (\url{http://konect.uni-koblenz. de/networks/facebook-wosn-wall}), etc. They occur in small proportions and can be easily accommodated by a slight modification in the model fitting procedure. The new model has parameters $(\alpha,\beta,\gamma,\xi,\din,\dout)$, and $\rho$ is implicitly defined through $\rho = 1- (\alpha+\beta+\gamma+\xi)$. Similar to the derivations in Section \ref{sec:estMLE}, the MLE estimators for $\alpha,\beta,\gamma,\xi$ are
\begin{align*}
\hat{\alpha}^{MLE} = \frac{1}{n} \sum_{t=1}^n \ind_{\{J_t=1\}}, & \quad \hat{\beta}^{MLE} = \frac{1}{n} \sum_{t=1}^n \ind_{\{J_t=2\}},\\
\hat{\gamma}^{MLE} = \frac{1}{n} \sum_{t=1}^n \ind_{\{J_t=3\}}, & \quad \hat{\xi}^{MLE} = \frac{1}{n} \sum_{t=1}^n \ind_{\{J_t=4\}},
\end{align*}
and $\din,\dout$ can be obtained through solving
\begin{align*}
\sum_{i=0}^\infty \frac{\Nin_{>i}(n)/n}{i+\din}-\frac{\frac{1}{n}\sum_{t=1}^n \ind_{\{J_t\in\{3,4,5\}\}}}{\din}
-\frac{1}{n}\sum_{t=1}^n \frac{N(t)}{t+\din N(t)}\ind_{\{J_t\in\lbrace 1, 2\rbrace\}}&=0, \\
\sum_{j=0}^\infty \frac{\Nout_{>j}(n)/n}{j+\dout}-\frac{\frac{1}{n}\sum_{t=1}^n \ind_{\{J_t\in\{1,4,5\}\}}}{\dout}
-\frac{1}{n}\sum_{t=1}^n \frac{N(t)}{t+\dout N(t)}\ind_{\{J_t\in\lbrace 2,3 \rbrace\}}&=0.
\end{align*}

We first naively fit the linear preferential attachment model to the full network using MLE. The MLE estimators are 
\begin{align} \label{eq:wikiparBefore}
(\hat\alpha,\hat\beta,\hat\gamma,&\hat\xi,\hat\rho,\hatdin,\hatdout)= \\
& (3.08\times10^{-3}, 8.55\times10^{-1}, 1.39\times10^{-1}, 4.76\times10^{-5},3.06\times10^{-3}, 0.547, 0.134 ).\nonumber
\end{align}
To evaluate the goodness-of-fit, 20 network realizations were simulated from the fitted model. We overlaid the {empirical} in- and out-degree frequencies of the original network with that of the simulations. If the model fits the data well, the {degree frequencies} of the data should lie within the range formed by that of the simulations, {which gives an informal confidence region for the degree distributions}. From Figure \ref{fig:wikiDegConst0}, we see that while the data roughly agrees with the simulations in the out-degree frequencies, the deviation in the in-degree frequencies is noticeable.

\begin{figure}[t]
\includegraphics[width=5in]{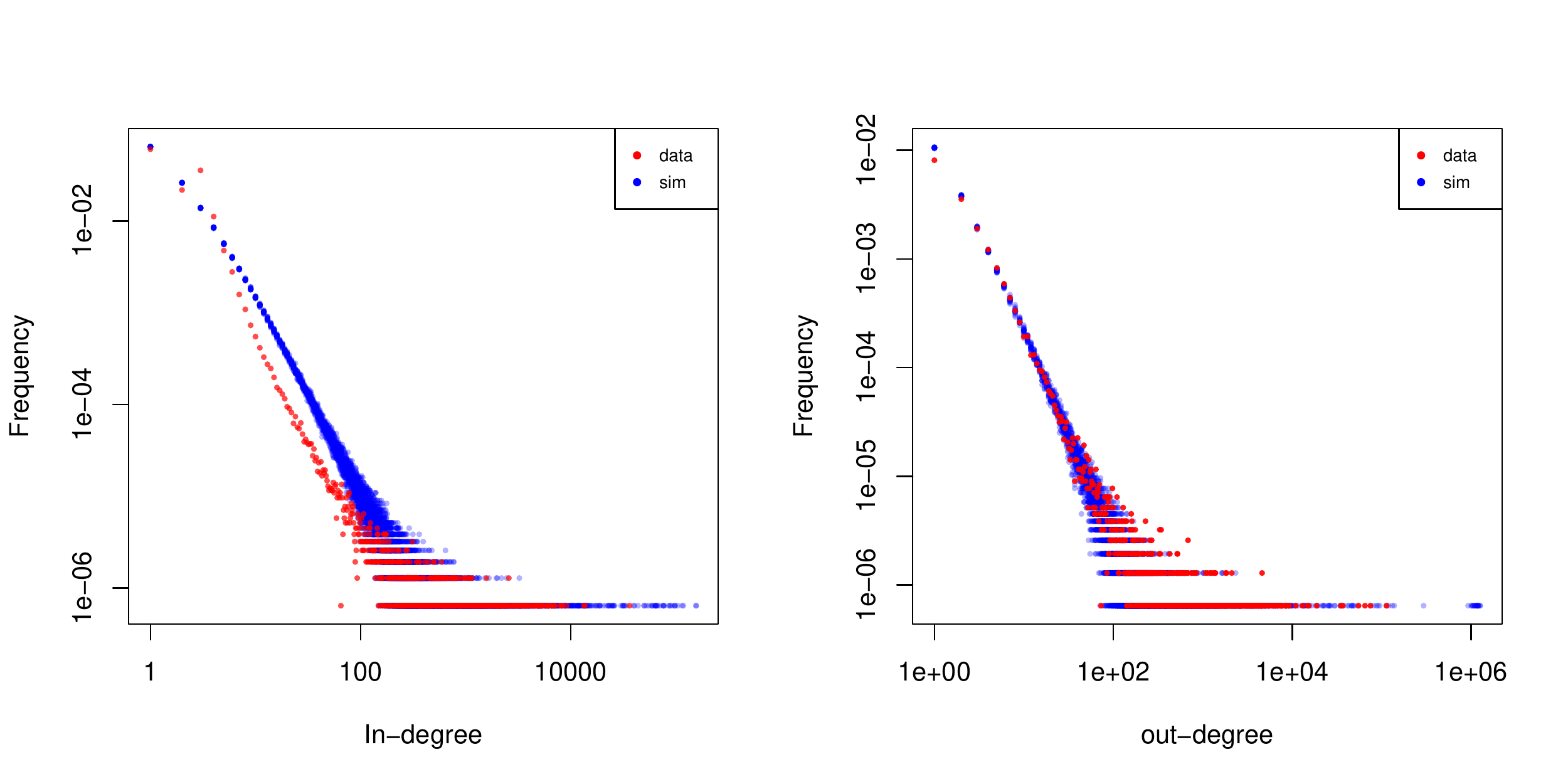}
\caption{{Empirical in- and out-degree frequencies of the full Wiki talk network (red) and that from 20 realizations of the linear preferential attachment network with fitted parameter values \eqref{eq:wikiparBefore} from MLE (blue). The scatter plots for the degree frequencies from the 20 simulations are overlaid together to form {an informal confidence region} for the degree distribution of the fitted model}}
\label{fig:wikiDegConst0}
\end{figure}

To better understand the discrepancy in the in-degree frequencies,
{we examined the link data and their timestamps} and discovered 
bursts of messages originating from certain nodes over small
time intervals. {According to Wikipedia policy \cite{wikipedia:2016}, certain administrating accounts are allowed to send group messages to multiple users simultaneously.} These bursts presumably represent broadcast announcements generated from these accounts.
These administrative broadcasts can {also be detected} if we
apply the linear preferential attachment model to the network in local
time intervals. We {divided the total time frame down to sub-intervals of varying length
each containing the formation of $10^4$ edges. The number $10^4$ is chosen to ensure good asymptotics as shown in Table~\ref{varyn}.
This process generated 155 networks,
$$
G(n_{k-1}),\dots, G(n_k-1), \quad k=1,\dots,155. 
$$
For each of the 155 datasets, we fit a
preferential attachment model using MLE.}
The resulting estimates $(\hatdin,\hatdout)$ are plotted against the
corresponding timeline on the upper left panel of Figure
\ref{fig:wikiParEst}. Notice that $\hatdin$ exhibits large spikes at
various times. Recall from \eqref{eq:probIn}, a large value of $\din$
indicates that the probability of an existing node $v$ receiving a new
message becomes less dependent on its in-degree, i.e., previous
popularity. These spikes appear to be directly related to the
occurrences of group messages. This plot is truncated after the day
2016/3/16, on which a massive group message of size 48,957 was sent
and the model can no longer be fit.

{We identified 37 users who have sent, at least once, 40 or
  more consecutive messages in the message history. This is evidence
  that  group messages were
 sent by this user. We presume these nodes are administrative
 accounts}; they are responsible for about $30\%$ of
the total messages sent. Since their behavior cannot be regarded as
normal social interaction, we excluded messages from these accounts
from the dataset in our analysis. We then also removed
nodes with zero in- and out-degrees. 

The re-estimated parameters after the data cleaning are displayed
in the other three panels of Figure
 \ref{fig:wikiParEst}. Here all parameter estimates are quite stable
 through time. 
\begin{figure}[t]
\includegraphics[width=5in]{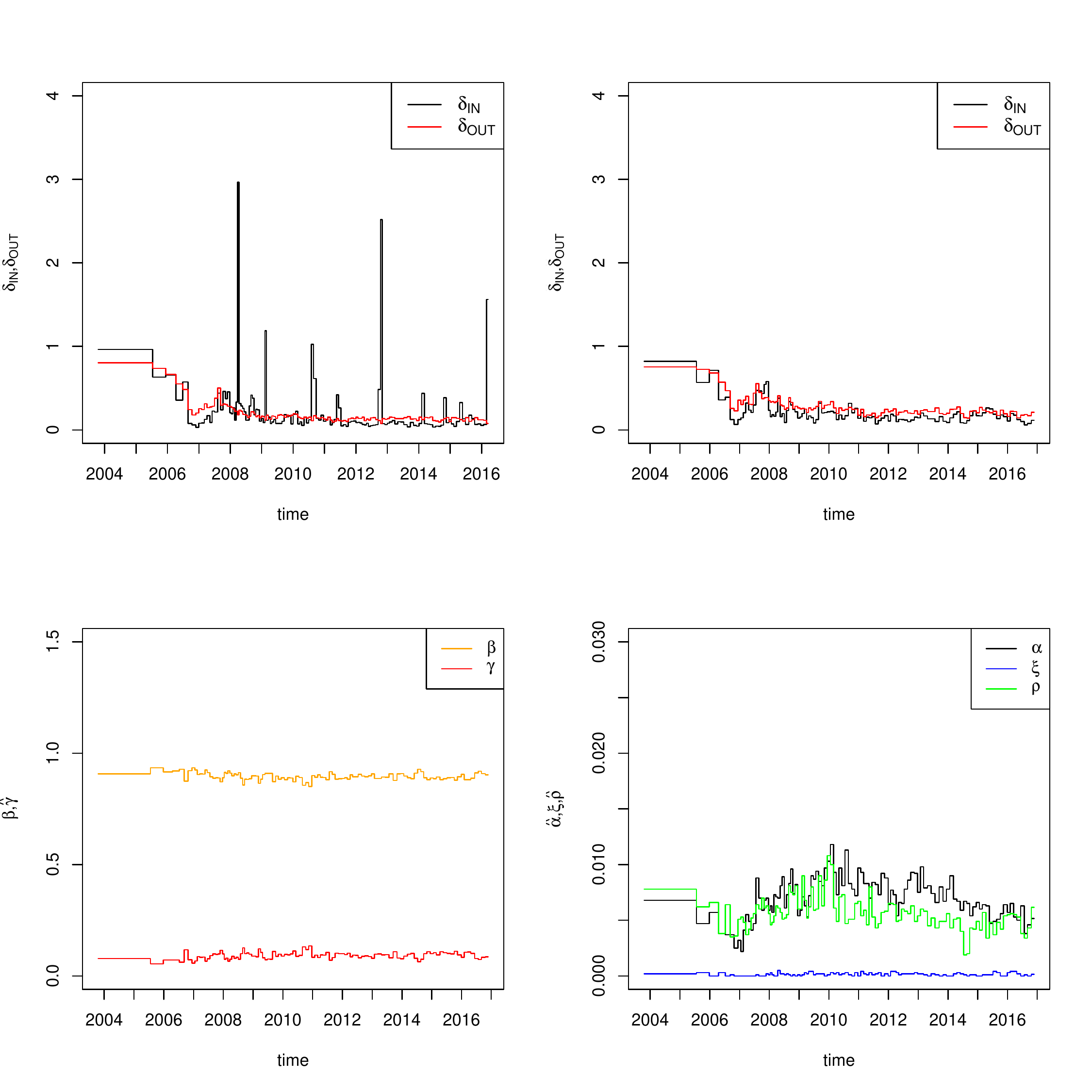}
\caption{Local parameter estimates of the linear preferential attachment model for the full and reduced Wiki talk network. Upper left: $(\hatdin,\hatdout)$ for the full network. Upper right, lower left, lower right: $(\hatdin,\hatdout)$, $(\hat\beta,\hat\gamma)$, $(\hat\alpha,\hat\xi,\hat\rho)$ for the reduced network, respectively.}
\label{fig:wikiParEst}
\end{figure}

The reduced network now contains 112,919 nodes and 1,086,982 edges, to
which we fit the linear preferential attachment model. The fitted
parameters based on MLE for our reduced dataset are  
\begin{align}\label{eq:wikiparAfter}
(\hat\alpha,\hat\beta,\hat\gamma,&\hat\xi,\hat\rho,\hatdin,\hatdout)=\\
&(6.95\times10^{-3}, 8.96\times10^{-1}, 9.10\times10^{-2}, 1.44\times10^{-4},5.61\times10^{-3}, 0.174, 0.257 ).\nonumber
\end{align}
Again the degree distributions of the data and 20 simulations from the
fitted model are displayed in Figure \ref{fig:wikiDegConst}. The
out-degree distribution of the data agrees reasonably well with the
simulations. For the in-degree distribution, the fit is better than that for the entire dataset (Figure \ref{fig:wikiDegConst0}). However, for smaller in-degrees, the fitted model over-estimates the in-degree frequencies. We speculate that in many social networks, the out-degree is {in line} with that predicted by the preferential attachment model. An individual node would be more likely to reach out to others if having done so many times previously. For in-degrees, the situation is complicated and may depend on a multitude of factors. For instance, the choice of recipient may
depend on the community that the sender is in, the topic being
discussed in the message, etc. As an example a group leader might send messages to his/her team on a regular basis. Such examples violate the base assumptions of the preferential attachment model and could result in the deviation between the data and the simulations.

\begin{figure}[t]
\includegraphics[width=5in]{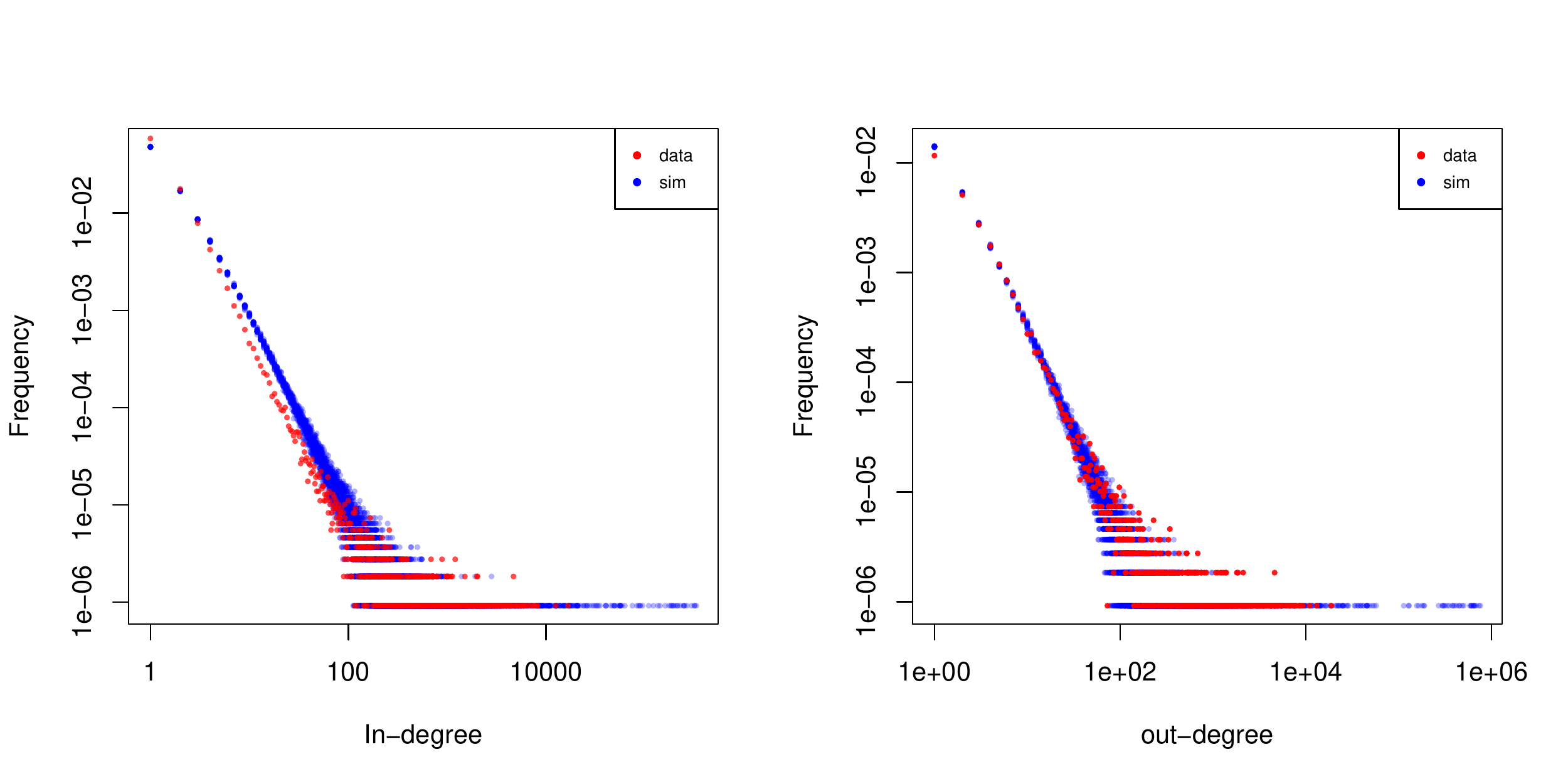}
\caption{{Empirical in- and out-degree frequencies of the reduced Wiki talk network (red) and that from 20 realizations of the linear preferential attachment network with fitted parameter values \eqref{eq:wikiparAfter} from MLE (blue). }}
\label{fig:wikiDegConst}
\end{figure}

{
Next we consider the estimation method of Section~\ref{OneSnapshot} applied to a single snapshot of the data.  In order to implement this procedure, we donned blinders and assumed that our dataset consists only of the information of the wiki data at the last timestamp.  That is, information about administrative broadcasts, and other aspects of the data learned by looking at the previous history of the data are unavailable.  In particular, we would have no knowledge of the existence of the two additional scenarios corresponding to $J_n=4, 5$.  With this in mind, we fit the three scenario model using the methods in Section~\ref{OneSnapshot}.  The fitted parameters are
 \beqq\label{eq:wikiparSnap}
(\tilde\alpha,\tilde\beta,\tilde\gamma,\tildedin,\tildedout)= 
(5.80\times10^{-4}, 8.55\times10^{-1}, 1.45\times10^{-1},0.199, 0.165 ).
\eeqq
The comparison of the degree distributions between the data and simulations from the fitted model is displayed in Figure~\ref{fig:wikiSnap} and is not too dissimilar to the plots in Figure~\ref{fig:wikiDegConst0} that are based on maximum likelihood estimation using the full network data.  In particular, the out-degree distribution is matched reasonably well, but the fitted model does a poor job of {capturing} the in-degree distribution.
}

\begin{figure}[t]
\includegraphics[width=5in]{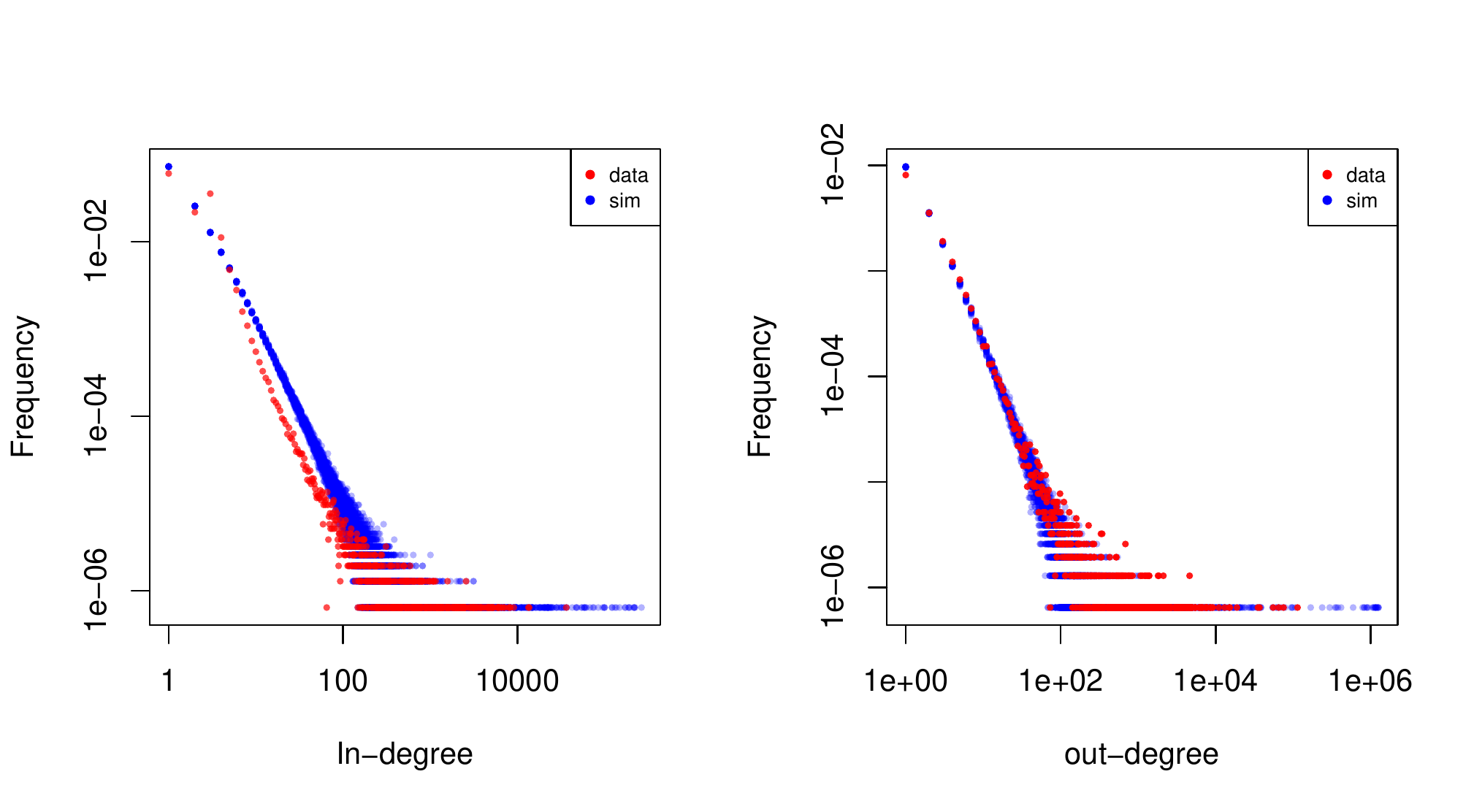}
\caption{{Empirical in- and out-degree frequencies of the full Wiki talk network (red) and that from 20 realizations of the linear preferential attachment network with fitted parameter values \eqref{eq:wikiparSnap} from the snapshot estimator (blue). }}
\label{fig:wikiSnap}
\end{figure}

We see from this example that while the linear preferential attachment model {is} perhaps too simplistic for
the Wiki talk network dataset, it has the ability to
{illuminate} some gross features, such as the out-degrees, as well as to
capture important structural changes such as the group message
behavior. {Consequently, despite its limitation, this model} may be used as a building block for more flexible models. Modification to the existing model formulation and
more careful analysis of change points in parameters is
 a direction for future research.

\section{Acknowledgement}
Research of the four authors was partially supported by Army MURI grant W911NF-12-1-0385. Don Towsley from University of Massachusetts introduced us to the model
and within his group, James Atwood graciously supplied us with a
simulation algorithm designed for a class of growth models broader
than the one specified in Section \ref{subsec:linpref}{; this later became \cite{atwood:2015}}. Joyjit Roy,
formerly of Cornell, created an efficient algorithm designed to
capitalize on the linear growth structure.
{Finally, we appreciate the many  helpful and sensible comments
of the referees and editors.}
\bibliography{./bibfile}


\appendix
\section{Proofs}
\label{sec:proofs}
\subsection{For the proof of Theorem~\ref{thm:consistency}: Lemmas~\ref{phi} and \ref{unifconv}}\label{subsec:proof1}
\begin{Lemma}\label{phi}
For ${ \lambda >0}$, the function $\psi(\lambda)$ in
\eqref{defpsi}
 has a unique zero at $\din$ and, $\psi(\lambda)>0$ when $\lambda <\din$ and $\psi(\lambda)<0$ when $\lambda >\din$.
\end{Lemma}


\begin{proof}
The probabilities $\{\pin_i(\lambda)\}$ satisfy the recursions in $i$ (cf.\ \cite{bollobas:borgs:chayes:riordan:2003}):
\begin{subequations}
\label{rec_pin} 
\begin{align}
\pin_0(\lambda) \left(\lambda + \frac{1}{a_1(\lambda)}\right) &= \frac{\alpha}{a_1(\lambda)},\label{rec_pin1}\\
\pin_1(\lambda) \left(1+\lambda + \frac{1}{a_1(\lambda)}\right) &= \lambda\pin_0(\lambda)+\frac{\gamma}{a_1(\lambda)},\\
\pin_2(\lambda) \left(2+\lambda + \frac{1}{a_1(\lambda)}\right) &= (1+\lambda)\pin_1(\lambda),\\
\vdots &\nonumber\\
\pin_i(\lambda) \left(i+\lambda + \frac{1}{a_1(\lambda)}\right) &=
                                                                  (i-1+\lambda)\pin_{i-1}(\lambda),
                                                                  \quad
                                                                 { (i\geq
                                                                  2),}
\end{align}
\end{subequations}
where $a_1(\lambda) := (\alpha + \beta) / (1+\lambda(1-\beta))$.
Summing {the recursions in \eqref{rec_pin}} from $0$ to $i$, we get (with the convention that $\sum_{i=0}^{-1}=0$)
$$
\sum_{k=0}^i \pin_k(\lambda)\left(k+\lambda + \frac{1}{a_1(\lambda)}\right)  = \sum_{k=0}^{i-1} (k+\lambda)\pin_{k}(\lambda) + \frac{\alpha}{a_1(\lambda)} + \frac{\gamma}{a_1(\lambda)}{\ind_{\{i\ge 1\}}}, \quad i\ge 0,
$$
which {can be simplified to}
\beqq\label{sum}
\frac{1}{a_1(\lambda)}\sum_{k=0}^i \pin_k(\lambda) + (i+\lambda)\pin_i(\lambda) = \frac{1-\beta}{a_1(\lambda)}-\frac{\gamma}{a_1(\lambda)}\ind_{\{i=0\}}, \quad i\ge 0.
\eeqq
From \eqref{pij},
\beqq \label{pisum}
\sum_{i=0}^\infty \pin_i(\lambda) = \sum_{i,j} p_{ij}(\lambda) = 1-\beta.
\eeqq
Hence by rearranging \eqref{sum}, we have
$$
(i+\lambda)\pin_i(\lambda) +\frac{\gamma}{a_1(\lambda)}\ind_{\{i=0\}} =\frac{1}{a_1(\lambda)}\left(1-\beta-\sum_{k=0}^i \pin_k(\lambda)\right) = \frac{1}{a_1(\lambda)}\pin_{>i}(\lambda),
$$
or equivalently,
\beqq\label{sumb}
\pin_{>i}(\lambda) = a_1(\lambda)(i+\lambda)\pin_i(\lambda) +\gamma\ind_{\{i=0\}}.
\eeqq
{Now with the help of \eqref{pisum} and \eqref{sumb},}
we can rewrite $\psi(\lambda)$ {in the following way:}
\begin{align}
\psi(\lambda) =&\ \sum_{i=0}^\infty\frac{\pin_{>i}(\din)}{i+\lambda} -\frac{\gamma}{\lambda} - (1-\beta)a_1(\lambda) \nonumber \\
=&\ \sum_{i=0}^\infty\frac{\pin_{>i}(\din)}{i+\lambda} -\frac{\gamma}{\lambda}-\sum_{i=0}^\infty \frac{\pin_i(\din)a_1(\lambda)(i+\lambda)}{i+\lambda} \nonumber \\
=&\ \sum_{i=0}^\infty\frac{a_1(\din)(i+\din)\pin_i(\din) +\gamma\ind_{\{i=0\}}}{i+\lambda} -\frac{\gamma}{\lambda}-\sum_{i=0}^\infty \frac{\pin_i(\din)a_1(\lambda)(i+\lambda)}{i+\lambda} \nonumber \\
 =&\ \sum_{i=0}^\infty\frac{ \pin_i(\din)}{i+\lambda}
    \Bigl(a_1(\din)(i+\din)- a_1(\lambda)(i+\lambda)\Bigr)
    \nonumber\\ 
 =&\ \sum_{i=0}^\infty\frac{\pin_i(\din) }{i+\lambda}  \int_\lambda^{\din} \frac{\partial}{\partial s} \Bigl(a_1(s)(i+s)\Bigr)ds  \nonumber\\
 =&\ \sum_{i=0}^\infty\frac{\pin_i(\din) }{i+\lambda}  \int_\lambda^{\din} \frac{(\alpha+\beta)(1-i(1-\beta))}{(1+s(1-\beta))^2} ds \nonumber\\
 =&\ \left(\sum_{i=0}^\infty\frac{\pin_i(\din)
    }{i+\lambda}(1-i(1-\beta))\right)  \int_\lambda^{\din}
    \frac{\alpha+\beta}{(1+s(1-\beta))^2} ds \nonumber \\
 =:&\ C(\lambda) \int_\lambda^{\din}
     \frac{\alpha+\beta}{(1+s(1-\beta))^2} ds. \label{extraGoodie}
\end{align}

{The series defining} $C(\lambda)$ converges absolutely for any $\lambda>0$ since
\begin{align*}
\sum_{i=0}^\infty\left|\frac{\pin_i(\din) }{i+\lambda}(1-i(1-\beta))\right| <\sum_{i=0}^\infty\pin_i(\din)\left|\frac{ i(1-\beta)}{i+\lambda} + \frac{1}{i+\lambda}\right| < (1-\beta) (1-\beta+\frac{1}{\lambda}) < \infty.
\end{align*}
{Summing over $i$ in $\eqref{sumb}$, we get by monotone convergence}
$$
\sum_{i=0}^\infty \pin_{>i}(\lambda) =\sum_{i=0}^\infty i\pin_i(\lambda)= a_1(\lambda) \sum_{i = 0}^\infty i\pin_i(\lambda) + a_1(\lambda) \lambda \sum_{i = 0}^\infty \pin_i(\lambda) +\gamma.
$$
{The infinite series converge because $\pin_i(\lambda)$ is a power law
with index greater than 2; see \eqref{asyI} and \eqref{c1}. Solving for the infinite
series we get}
\beqq\label{ipi_limit}
\sum_{i=0}^\infty i\pin_{i}(\lambda) = \frac{a_1(\lambda)\lambda}{1-a_1(\lambda)} (1-\beta) +\frac{\gamma}{1-a_1(\lambda)} = 1.
\eeqq
Hence we have
\begin{align*}
C(\lambda) =&\  \sum_{i\le (1-\beta)^{-1}}\frac{\pin_i(\din) }{i+\lambda}(1-i(1-\beta)) -  \sum_{i > (1-\beta)^{-1}}\frac{\pin_i(\din) }{i+\lambda}(i(1-\beta)-1)\\
>&\ \sum_{i=0}^\infty\frac{\pin_i(\din) }{(1-\beta)^{-1}+\lambda}(1-i(1-\beta)) \\
=&\ \frac{1}{(1-\beta)^{-1}+\lambda}\sum_{i=0}^\infty \pin_i(\din) -\frac{1-\beta}{(1-\beta)^{-1}+\lambda}\sum_{i=0}^\infty i\pin_i(\din) \\
=&\ \frac{1}{(1-\beta)^{-1}+\lambda}(1-\beta) - \frac{1-\beta}{(1-\beta)^{-1}+\lambda} 1 \\
=&\ 0.
\end{align*}

Now recall from \eqref{extraGoodie} that $\psi(\lambda)$ is of the form
$$
\psi(\lambda) =  C(\lambda)\int_\lambda^{\din} \frac{\alpha+\beta}{(1+s(1-\beta))^2} ds,
$$
where $C(\lambda)>0$ for all ${\lambda >0}$. Therefore $\psi(\cdot)$ has a unique zero at $\din$ and $\psi(\lambda)>0$ when $\lambda <\din$ and $\psi(\lambda)<0$ when $\lambda >\din$.\end{proof}

We show the uniform convergence {of}
 $\psi_n$ to $\psi$ in the next lemma.

\begin{Lemma}\label{unifconv}
As $n\to\infty$, for any $\epsilon >0$,
$$
\sup_{\lambda \geq \epsilon}|\psi_n(\lambda)-\psi(\lambda)|\convas 0.
$$
\end{Lemma}


\begin{proof}
By the definition of $\psi$, $\pin_{>i}(\din)$ {is a function of $\din$ and is a constant with respect to $\lambda$}. Hence we {suppress the dependence {on $\din$} and simply write it as $\pin_{>i}$} when considering the difference $\psi_n-\psi$ as a function of $\lambda$:
\begin{align*}
\psi_n(\lambda) -\psi(\lambda) 
=&   \sum_{i=0}^\infty \frac{\Nin_{>i}(n)/n-\pin_{>i}}{i+\lambda}
- \frac{1}{\lambda}\left(\frac{1}{n}\sum_{t=1}^n \ind_{\{J_t=3\}}-(1-\alpha-\beta)\right)\\
&- \frac{1}{n}\sum_{t=1}^n \left(\frac{N(t-1)}{t-1+\lambda N(t-1)}\ind_{\{J_t\in\lbrace 1,2 \rbrace\}}-\frac{(1-\beta)(\alpha+\beta)}{1+\lambda(1-\beta)}\right).
\end{align*}
Thus,
\begin{align}\label{diff}
\sup_{\lambda{\ge\epsilon}}|\psi_n(\lambda)-\psi(\lambda)|
\le & \sup_{\lambda{\ge\epsilon}}\sum_{i=0}^\infty \frac{\left\vert\Nin_{>i}(n)/n-\pin_{>i}\right\vert}{i+\lambda}
+ \sup_{\lambda{\ge\epsilon}}\frac{1}{\lambda}\left|\frac{1}{n}\sum_{t=1}^n \ind_{\{J_t=3\}}-(1-\alpha-\beta)\right|\nonumber\\
&+ \sup_{\lambda{\ge\epsilon}}\left\vert\frac{1}{n}\sum_{t=1}^n \frac{N(t-1)}{t-1+\lambda N(t-1)}\ind_{\{J_t\in\lbrace 1,2 \rbrace\}}-\frac{(1-\beta)(\alpha+\beta)}{1+\lambda(1-\beta)}\right\vert.
\end{align}

For the first term, note that for all $i\ge 0$, 
\[
i\Nin_{>i}(n) = \sum_{k=i+1}^\infty \Nin_k(n) i \le \sum_{k=1}^\infty k\Nin_k(n) = n,
\]
since the assumption on initial conditions implies the sum of in-degrees at $n$ is $n$. Therefore $\Nin_{>i}(n)/n\le i^{-1}$ for $i\ge 1$,
and it then follows that
\[
\sum_{i=0}^\infty \frac{\left\vert\Nin_{>i}(n)/n-\pin_{>i}\right\vert}{i+\lambda}
\le \sum_{i=0}^M \frac{\left\vert\Nin_{>i}(n)/n-\pin_{>i}\right\vert}{i+\lambda}
+ \sum_{i=M+1}^\infty \frac{1/i}{i+\lambda} + \sum_{i=M+1}^\infty\frac{\pin_{>i}}{i+\lambda}.
\]
Note that the last two terms on the right side
can be made arbitrarily small uniformly on $[\epsilon,{\infty)}$ if we choose $M$ sufficiently large.

Recall the convergence of the degree distribution $\{N_{ij}(n)/N(n)\}$ to the probability distribution $\{f_{ij}\}$ in \eqref{pij}, we have
\beqq \label{eq:Nin_over}
\frac{\Nin_{>i}(n)}{n} = \frac{N(n)}{n}\ \frac{\Nin_{>i}(n)}{N(n)} \ \convas\  (1-\beta)\sum_{l\ge0,k>i}f_{kl}  = \pin_{>i},\quad \forall i\ge 0.
\eeqq
Hence, for any fixed $M$, 
\[
\sum_{i=0}^M \frac{\left\vert\Nin_{>i}(n)/n-\pin_{>i}\right\vert}{i+\epsilon}\convas 0,\quad \mbox{as }n\to\infty.
\]
which implies further that choosing $M$ arbitrarily large gives
\[
\sup_{\lambda{\ge\epsilon}}\sum_{i=0}^\infty \frac{\left\vert\Nin_{>i}(n)/n-\pin_{>i}\right\vert}{i+\lambda}
\le \sum_{i=0}^M \frac{\left\vert\Nin_{>i}(n)/n-\pin_{>i}\right\vert}{i+\epsilon}
+ \sum_{i=M+1}^\infty \frac{1/i}{i+\epsilon} 
+ \sum_{i=M+1}^\infty\frac{\pin_{>i}}{i+\epsilon}\stackrel{\text{a.s.}}{\longrightarrow} 0.
\]

The second term in \eqref{diff} converges to 0 almost surely by strong law of large numbers, and
the third term in \eqref{diff} can be written as
\begin{align*}
\left\vert \frac{1}{n}\sum_{t=1}^n \right. & \left. \left(\frac{N(t-1)}{t-1+\lambda N(t-1)}-\frac{(1-\beta)}{1+\lambda(1-\beta)}\right)\ind_{\{J_t\in\lbrace 1,2 \rbrace\}}\right.\\
&\left.+\frac{1-\beta}{1+\lambda(1-\beta)}\frac{1}{n}\sum_{t=1}^n\left(\ind_{\{J_t\in\lbrace 1,2 \rbrace\}}-(\alpha+\beta)\right)\right\vert,
\end{align*}
which is bounded by 
\[
\left\vert\frac{1}{n}\sum_{t=1}^n \frac{N(t-1)}{t-1+\lambda N(t-1)}-\frac{(1-\beta)}{1+\lambda(1-\beta)}\right\vert
+ \frac{1-\beta}{1+\lambda(1-\beta)}\left\vert\frac{1}{n}\sum_{t=1}^n\ind_{\{J_t\in\lbrace 1,2 \rbrace\}}-(\alpha+\beta)\right\vert.
\]
We have
\begin{align*}
 \sup_{\lambda{\ge\epsilon}}\Bigl\vert\frac{1}{n}\sum_{t=1}^n &\frac{N(t-1)}{t-1+\lambda N(t-1)}-\frac{(1-\beta)}{1+\lambda(1-\beta)}\Bigr\vert \\
&= \sup_{\lambda{\ge\epsilon}}\left\vert \frac{1}{n}\sum_{t=1}^n\frac{N(t-1)/(t-1)-(1-\beta)}{(1+\lambda N(t-1)/(t-1))(1+\lambda(1-\beta))} \right\vert\\
&\le \frac{1}{n}\sum_{t=1}^n\left\vert \frac{N(t-1)/(t-1)-(1-\beta)}{(1+\epsilon N(t-1)/(t-1))(1+\epsilon(1-\beta))}\right\vert,
\end{align*}
which converges to 0 almost surely by Ces\`aro convergence of random variables, since
\[
\left\vert \frac{N(n)/n-(1-\beta)}{(1+\epsilon N(n)/n)(1+\epsilon(1-\beta))}\right\vert\convas 0,\, \mbox{ as }n\to\infty.
\]
Further, by the strong law of large numbers,
\begin{align*}
\sup_{\lambda{\ge\epsilon}}\ & \frac{1-\beta}{1+\lambda(1-\beta)}
\left\vert\frac{1}{n}\sum_{t=1}^n\ind_{\{J_t\in\lbrace 1,2 \rbrace\}}-(\alpha+\beta)\right\vert \\
&\le \frac{1-\beta}{1+\epsilon(1-\beta)}\left\vert\frac{1}{n}\sum_{t=1}^n\ind_{\{J_t\in\lbrace 1,2 \rbrace\}}-(\alpha+\beta)\right\vert
\convas 0, \, \mbox{ as }n\to\infty.
\end{align*}
Hence the third term of \eqref{diff} also goes to 0 almost surely as $n\to\infty$. The result of the lemma follows.
\end{proof}

\subsection{For the proof of Theorem~\ref{asymp_normality}: Lemmas~\ref{normality_lemma2} and \ref{normality_lemma3}}\label{subsec:proof2}

\begin{Lemma} \label{normality_lemma2}
As $n\to\infty$,
\beqq \label{score_conv}
n^{-1/2}  \sum_{t=1}^n u_t(\din) \convd N(0,I_\text{in}).
\eeqq
\end{Lemma}

\begin{proof}
Let $\mathcal{F}_{n}=\sigma(G(0), \ldots, G(n))$ be the $\sigma$-field generated by the
information contained in the graphs.
We first observe that {$\{ \sum_{t=1}^n
u_t(\din),\mathcal{F}_{n}, n\geq 1\}$}
is a martingale. To see this, note from \eqref{ut_def} that
$|u_t(\delta)| \le 2/\delta$ and  
\begin{align*}
\EE[u_t&(\din)|\mathcal{F}_{t-1}] \\
=&\ \EE\left[\left.\frac{1}{\Din^{(t-1)}(\vend_t)+\din}\ind_{\{J_t\in\lbrace 1, 2\rbrace\}}\right|\mathcal{F}_{t-1}\right] - \frac{N(t-1)}{t-1+\din N(t-1)} \EE[
\ind_{\{J_t\in\lbrace 1, 2\rbrace\}}|\mathcal{F}_{t-1}]
\\
=&\ \EE\left[\left.\frac{1}{\Din^{(t-1)}(\vend_t)+\din}\right|J_t=1,\mathcal{F}_{t-1}\right] \PP[J_t=1] \\
&\ + \EE\left[\left.\frac{1}{\Din^{(t-1)}(\vend_t)+\din}\right|J_t=2,\mathcal{F}_{t-1}\right] \PP[J_t=2] -  (\alpha+\beta) \frac{N(t-1)}{t-1+\din N(t-1)} \\
=&\ (\alpha+\beta)\sum_{v\in V_{t-1}}\frac{1}{\Din^{(t-1)}(v)+\din} \frac{\Din^{(t-1)}(v)+\din} {t-1+\din N(t-1)} - (\alpha+\beta)\frac{N(t-1)}{t-1+\din N(t-1)}\\
=&\ (\alpha+\beta) \left( \sum_{v\in V_{t-1}}\frac{1}{t-1+\din N(t-1)} -\frac{N(t-1)}{t-1+\din N(t-1)}\right)\\
=&\ 0,
\end{align*}
which satisfies the definition of a martingale difference. Hence $\left\{n^{-1/2}  \sum_{r=1}^t u_r(\din)\right\}_{t=1,\ldots,n}$ is a zero-mean, square-integrable martingale array. The convergence \eqref{score_conv} follows from {the martingale central limit theory (cf.~Theorem 3.2 of \cite{hall:heyde:1980})} if the following three conditions can be verified:
\begin{enumerate}
\item[(a)]
	$n^{-1/2}\max_t  |u_t(\din)| \convp 0$,
\item[(b)]
	$ n^{-1} \sum_t u_t^2(\din) \convp I_\text{in}$,
\item[(c)]
	$\EE \left(n^{-1} \max_t  u_t^2(\din)\right)$ is bounded in $n$.
\end{enumerate}

Since $|u_t(\din)| \le 2/\din$, we have
$$
n^{-1/2}\max_t | u_t(\din)| \le \frac{2}{n^{1/2}\din} \to 0,
$$
and
$$
n^{-1}\max_t  u_t^2 \le \frac{4}{n\din^2} \to 0.
$$
Hence conditions (a) and (c) are straightforward.

To show (b), observe that
\begin{align*}
\frac{1}{n}\sum_{t=1}^n u_t^2(\din) =&\ \frac{1}{n}\sum_{t=1}^n\ \ind_{\{J_t\in\lbrace 1, 2\rbrace\}} \left(\frac{1}{\Din^{(t-1)}(\vend_t)+\din} -  \frac{N(t-1)}{t-1+\din N(t-1)}\right)^2 \\
=&\ \frac{1}{n}\sum_{t=1}^n \frac{\ind_{\{J_t\in\lbrace 1, 2\rbrace\}}}{\left(\Din^{(t-1)}(\vend_t)+\din\right)^2} -  \frac{2}{n}\sum_{t=1}^n \frac{\ind_{\{J_t\in\lbrace 1, 2\rbrace\}}}{\Din^{(t-1)}(\vend_t)+\din} \frac{N(t-1)}{t-1+\din N(t-1)} \\
&\ + \frac{1}{n}\sum_{t=1}^n \ind_{\{J_t\in\lbrace 1, 2\rbrace\}}\left( \frac{N(t-1)}{t-1+\din N(t-1)}\right)^2 \\
=&: \ T_1 - 2 T_2 + T_3.
\end{align*}
Following the calculations in the proof of Lemma~\ref{unifconv},
we have for $T_1$,
\begin{align*}
T_1 =&  \sum_{i=0}^\infty \frac{\Nin_{>i}(n)/n}{(i+\din)^2} - \frac{1}{\din^2} \frac{1}{n}\sum_{t=1}^{n}\ind_{\{J_t=3\}}
\convp \sum_{i=0}^\infty \frac{\pin_{>i}}{(i+\din)^2} - \frac{\gamma}{\din^2}.
\end{align*}
We then rewrite $T_2$ as
\begin{align*}
T_2 =&\ \frac{1}{n}\sum_{t=1}^n \frac{\ind_{\{J_t\in\lbrace 1, 2\rbrace\}}}{\Din^{(t-1)}(\vend_t)+\din} 
\left(\frac{N(t-1)/(t-1)}{1+\din N(t-1)/(t-1)} - \frac{1-\beta}{1+\din(1-\beta)} \right) \\
&\ + \frac{1}{n}\sum_{t=1}^n \frac{\ind_{\{J_t\in\lbrace 1, 2\rbrace\}}}{\Din^{(t-1)}(\vend_t)+\din} \frac{1-\beta}{1+\din(1-\beta)} \\
=&:\ T_{21} + T_{22},
\end{align*}
where 
$$
|T_{21}| \le \frac{1}{n}\sum_{t=1}^n \frac{1}{\din} 
\left|\frac{N(t-1)/(t-1)}{1+\din N(t-1)/(t-1)}- \frac{1-\beta}{1+\din(1-\beta)}\right| \convp 0
$$
by Ces\`aro's convergence and
\begin{align*}
T_{22} &= \frac{1-\beta}{1+\din(1-\beta)} \left(\sum_{i=0}^\infty \frac{\Nin_{>i}(n)/n}{i+\din} - \frac{1}{\din} \frac{1}{n}\sum_{t=1}^{n}\ind_{\{J_t=3\}}\right)\\
&\convp \frac{1-\beta}{1+\din(1-\beta)}\left(\sum_{i=0}^\infty \frac{\pin_{>i}}{i+\din} - \frac{\gamma}{\din}\right)
=  \frac{(\alpha+\beta)(1-\beta)^2}{(1+\din(1-\beta))^2},
\end{align*}
where the equality follows from \eqref{sumb}.
For $T_3$, similar to $T_1$, we have
\begin{align*}
T_3 =&  \frac{1}{n}\sum_{t=1}^n \ind_{\{J_t\in\lbrace 1, 2\rbrace\}}\left(\left( \frac{N(t-1)/(t-1)}{1+\din N(t-1)/(t-1)}\right)^2 - \frac{(1-\beta)^2}{(1+\din(1-\beta))^2}\right)\\
&+ \frac{(1-\beta)^2}{(1+\din(1-\beta))^2} \frac{1}{n}\sum_{t=1}^n \ind_{\{J_t\in\lbrace 1, 2\rbrace\}}
\convp \frac{(\alpha+\beta) (1-\beta)^2}{(1+\din(1-\beta))^2}.
\end{align*}

Combining these results together,
\begin{align}
\frac{1}{n}\sum_{t=1}^n u_t^2(\din) =&\ T_1 - 2 (T_{21} + T_{22}) + T_3 \nonumber\\
\convp&\  \sum_{i=0}^\infty \frac{\pin_{>i}}{(i+\din)^2} - \frac{\gamma}{\din^2}- \frac{(\alpha+\beta)(1-\beta)^2}{(1+\din(1-\beta))^2}
=\ I_\text{in} \label{conv_ut_sq}.
\end{align}
This completes the proof.
\end{proof}

\begin{Lemma} \label{normality_lemma3}
As $n\to\infty$,
$$\frac{1}{n}\sum_{t=1}^n \dot{u}_t(\hatdin^*) \convp -I_\text{in}.$$
\end{Lemma}

\begin{proof}
The result of this lemma can be established by showing first 
\beqq \label{conv_udot1}
\frac{1}{n}\sum_{t=1}^n \dot{u}_t(\din) \convp -I_\text{in}
\eeqq
and then
\beqq \label{conv_udot2}
\left|\frac{1}{n}\sum_{t=1}^n \dot{u}_t(\hatdin^*) - \frac{1}{n}\sum_{t=1}^n \dot{u}_t(\din)\right| \convp 0.
\eeqq

We first observe that
\begin{align*}
\dot{u}_t(\delta) =&\ -\left(\frac{1}{\Din^{(t-1)}(\vend_t)+\delta}\right)^2 \ind_{\{J_t\in\lbrace 1, 2\rbrace\}}
+ \left(\frac{N(t-1)}{t-1+\delta N(t-1)}\right)^2\ind_{\{J_t\in\lbrace 1, 2\rbrace\}} \\
=&\ -u_t^2(\delta) - 2u_t(\delta)  \frac{N(t-1)}{t-1+\delta N(t-1)}.
\end{align*}
Recall the definition and convergence result for $T_2$ and $T_3$ in Lemma~\ref{normality_lemma2}, we have
\begin{align*}
\frac{1}{n}\sum_{t=1}^n  u_t(\din)  \frac{N(t-1)}{t-1+{\din} N(t-1)}=&\ T_2-T_3 \convp 0.
\end{align*}
Also from \eqref{conv_ut_sq},
$$
\frac{1}{n}\sum_{t=1}^n u^2_t(\din) \convp I_\text{in}.
$$
Hence
$$
\frac{1}{n}\sum_{t=1}^n \dot{u}_t(\din) = -\frac{1}{n}\sum_{t=1}^nu_t^2(\din) - \frac{2}{n}\sum_{t=1}^nu_t(\din)  \frac{N(t-1)}{t-1+{\din} N(t-1)} \convp -I_\text{in}
$$
and \eqref{conv_udot1} is established.

By construction and definition, we have $\hatdin,\hatdin^*,\din>0$. To prove \eqref{conv_udot2}, note that
\begin{align*}
|u_t(\hatdin^*) - u_t(\din)| \le&\ \ind_{\{J_t\in\lbrace 1, 2\rbrace\}} \left|\frac{1}{\Din^{(t-1)}(\vend_t)+\hatdin^*}  - \frac{1}{\Din^{(t-1)}(\vend_t)+\din}\right| \\
&\ \quad +  \ind_{\{J_t\in\lbrace 1, 2\rbrace\}} \left|\frac{N(t-1)}{t-1+\hatdin^* N(t-1)} -\frac{N(t-1)}{t-1+\din N(t-1)}\right| \\
=&\ \ind_{\{J_t\in\lbrace 1, 2\rbrace\}} \left|\frac{\din-\hatdin^*}{\left(\Din^{(t-1)}(\vend_t)+\hatdin^*\right)\left(\Din^{(t-1)}(\vend_t)+\din\right)}\right| \\
&\ \quad +  \ind_{\{J_t\in\lbrace 1, 2\rbrace\}} \left|\frac{(N(t-1))^2(\din-\hatdin^*)}{\left(t-1+\hatdin^* N(t-1)\right)\left(t-1+\din N(t-1)\right)}\right| \\
\le &\ \frac{2|\hatdin^* - \din|}{\hatdin^*\din}.
\end{align*}
Then
\begin{align*}
|u_t^2(\hatdin^*) - u_t^2(\din)| =\ \left|u_t(\hatdin^*) - u_t(\din)\right|\left|u_t(\hatdin^*) + u_t(\din)\right| 
\le &\ \frac{2\left|\hatdin^* - \din\right|}{\hatdin^*\din}\left(\frac{2}{\hatdin^*} + \frac{2}{\din}\right),
\end{align*}
and
\begin{align*}
\left| u_t(\hatdin^*)\right.&\left.\frac{N(t-1)}{t-1+\hatdin^* N(t-1)}
 - u_t(\din)\frac{N(t-1)}{t-1+\din N(t-1)}\right|\\
  \le&\ \left|u_t(\hatdin^*)
 - u_t(\din)\right|\frac{\frac{N(t-1)}{t-1}}{1+\din \frac{N(t-1)}{t-1}}  + \left|u_t(\hatdin^*)\right|\left|\frac{\frac{N(t-1)}{t-1}}{1+\hatdin^* \frac{N(t-1)}{t-1}} - \frac{\frac{N(t-1)}{t-1}}{1+\din \frac{N(t-1)}{t-1}}\right| \\
 \le &\ \frac{2\left|\hatdin^* - \din\right|}{\hatdin^*\din}\frac{1}{\din} + \frac{2}{\hatdin^*}\frac{\left|\hatdin^* - \din\right|}{\hatdin^*\din}.
\end{align*}
From Theorem~\ref{thm:consistency}, $\hatdin^{MLE}$ is consistent for $\din$, hence
$$
\left|\hatdin^*-\din\right| \le \left|\hatdin^{MLE} - \din\right| \convp 0.
$$
We have
\begin{align*}
\left| \frac{1}{n}\sum_{t=1}^n \right. & \left. \dot{u}_t(\hatdin^*) - \frac{1}{n}\sum_{t=1}^n \dot{u}_t(\din)\right|\\ 
\le &\ 
\frac{1}{n}\sum_{t=1}^n \left|\dot{u}_t(\hatdin^*) - \dot{u}_t(\din)\right| 
\le \ \frac{1}{n}\sum_{t=1}^n \left|{u}^2_t(\hatdin^*) - {u}^2_t(\din)\right| \\
&\ + \frac{2}{n}\sum_{t=1}^n \left|u_t(\hatdin^*)\frac{N(t-1)}{t-1+\hatdin^* N(t-1)} - u_t(\din)\frac{N(t-1)}{t-1+\din N(t-1)}\right| \\
\le &\ \frac{2\left|\hatdin^* - \din\right|}{\hatdin^*\din}\left(\frac{2}{\hatdin^*} + \frac{2}{\din}\right) + \frac{4\left|\hatdin^* - \din\right|}{\hatdin^*\din}\frac{1}{\din} + \frac{4}{\hatdin^*}\frac{\left|\hatdin^* - \din\right|}{\hatdin^*\din} 
\convp\ 0.
\end{align*}
This proves \eqref{conv_udot2} and completes the proof of Lemma~\ref{normality_lemma3}.
\end{proof}


\subsection{Proof of Theorem~\ref{ss_consist}}\label{subsec:proof3}
\begin{proof}
First observe that $\sum_i i \Nin_{i}(n)$ sums up to the total number of edges $n$, so
\[
\sum_{i=0}^\infty\frac{\Nin_{>i}(n)}{n} = \sum_{i=0}^\infty\frac{i\Nin_{i}(n)}{n} = 1.
\]
We can re-write \eqref{onesnap1} as
\begin{align}
\alpha + \tilde\beta\ &= \  \left(\frac{1}{\din}-\sum_{i=0}^\infty \frac{\Nin_{>i}(n)/n}{i+\din}\right)\bigg/\left(\frac{1}{\din} - \frac{1-\tilde\beta}{1+\din(1-\tilde\beta)}\right)  \nonumber \\
&= \  \left(\sum_{i=0}^\infty \frac{\Nin_{>i}(n)/n}{\din}-\sum_{i=0}^\infty \frac{\Nin_{>i}(n)/n}{i+\din}\right)\bigg/\left(\frac{1}{\din(1+\din(1-\tilde\beta))}\right) \nonumber \\
&=\ \sum_{i=1}^\infty \frac{\Nin_{>i}(n)}{n}\frac{i}{i+\din}\left(1+\din(1-\tilde\beta)\right)
=:\ f_n(\din) \label{fndef},
\end{align}
and \eqref{onesnap2} as
\begin{align}
\alpha +\tilde\beta\ &= \   \left(\frac{\Nin_0(n)}{n} + \tilde\beta \right) \bigg/ \left(1-\frac{\Nin_0(n)}{n} \frac{\din}{1+(1-\tilde\beta)\din}\right) 
=:\ g_n(\din). \nonumber
\end{align}

Then $\tildedin$ can be obtained by solving
$$
 f_n(\delta) - g_n(\delta) =0 ,\qquad \delta\in[\epsilon, K].
$$
Similar to the proof of Theorem~\ref{thm:consistency}, we define the limit versions of $f_n$, and $g_n$ as follows: 
\begin{align*}
f(\delta) := &\ \sum_{i=1}^\infty \pin_{>i}\frac{i}{i+\delta}(1+\delta(1-\beta)),   \\
g(\delta):= & \ \left(\pin_0 + \beta \right) \bigg/ \left(1-\pin_0 \frac{\delta}{1+(1-\beta)\delta}\right),
\qquad \delta\in[\epsilon,K].
\end{align*}
Now we apply the re-parametrization 
\beqq\label{reparam}
\eta := \frac{\delta}{1+\delta(1-\beta)}\in\left[\frac{1}{\epsilon^{-1}+1-\beta},\, \frac{1}{K^{-1}+1-\beta}\right] =:\mathcal{I}
\eeqq
to $f$ and $g$, such that
\begin{align*}
\tilde f(\eta) := &\ f(\delta(\eta)) =  \sum_{i=1}^\infty \frac{ \pin_{>i}}{1+(i^{-1}-(1-\beta))\eta}, \\
\tilde g(\eta) := &\ g(\delta(\eta)) = \frac{\pin_0 + \beta}{1 - \eta\pin_0}.
\end{align*}
Note that for all 
$\eta\in\mathcal{I}$:
\begin{itemize}
\item Set $b_i(\eta):=(i^{-1}-(1-\beta))\eta$, then $1+b_i(\eta) >0$ for all $i\ge 1$. So $\tilde{f}(\eta)>0$ on $\mathcal{I}$;
\item $\tilde{f} (\eta) \le \frac{1}{1-(1-\beta)\eta} \sum_{i = 0}^\infty\pin_{>i} \le 1+(1-\beta)K < \infty$.
\end{itemize}
Meanwhile, $\tilde{g}$ is also well defined and strictly positive for $\eta\in\mathcal{I}$ because 
\beqq\label{gfcn}
1/\pin_0 > 1/(1-\beta) > \eta.
\eeqq
The first inequality holds since:
\begin{align*}
1/\pin_0 > 1/(1-\beta) & \Leftrightarrow \pin_0 < 1-\beta \\
& \Leftrightarrow \frac{\alpha}{1+\frac{(\alpha+\beta)\din}{1+(1-\beta)\din}} < 1-\beta \\
& \Leftrightarrow \alpha+\beta < 1+ \frac{(1-\beta)(\alpha+\beta)\din}{1+(1-\beta)\din} \\
& \Leftrightarrow \alpha+\beta < 1+(1-\beta)\din.
\end{align*}
We know $\alpha+\beta<1$ by our model assumption, thus verifying \eqref{gfcn}.

Define for $\eta\in\mathcal{I}$,
$$
\tilde h (\eta) := \frac{1}{\tilde f(\eta)} - \frac{1}{\tilde g(\eta)} = \left( \sum_{i=1}^\infty \frac{ \pin_{>i}}{1+(i^{-1}-(1-\beta))\eta}\right)^{-1} - \frac{1 - \eta\pin_0}{\pin_0 + \beta},
$$
then it follows that 
\[
\tilde{h}(\eta) = 0 \quad\Leftrightarrow \quad \tilde{f}(\eta) = \tilde{g}(\eta) ,\qquad \eta\in\mathcal{I}.
\]
We now show that $\tilde{h}$ is concave and $\tilde{h}(\eta)\to 0$ as $\eta\to 0$, then the uniqueness of the solution follows.

First observe that 
\begin{align}
&\quad\quad\frac{\partial^2}{\partial \eta^2}\tilde h(\eta) 
= \frac{\partial^2}{\partial \eta^2} \left( \sum_{i=1}^\infty \frac{ \pin_{>i}}{1+(i^{-1}-(1-\beta))\eta}\right)^{-1}
= \frac{\partial^2}{\partial \eta^2} \left( \sum_{i=1}^\infty \frac{ \pin_{>i}}{1+b_i(\eta)}\right)^{-1}\nonumber\\
&= 2 \left(\sum_{i=1}^\infty\frac{\pin_{>i}}{1+b_i(\eta)}\right)^{-3}\left[\frac{\partial}{\partial \eta}\left(\sum_{i=1}^\infty\frac{\pin_{>i}}{1+b_i(\eta)}\right)\right]^2
- \left(\sum_{i=1}^\infty\frac{\pin_{>i}}{1+b_i(\eta)}\right)^{-2}\frac{\partial^2}{\partial \eta^2}\left(\sum_{i=1}^\infty\frac{\pin_{>i}}{1+b_i(\eta)}\right).\label{tildeh}
\end{align}
We now claim that 
\begin{align}
\frac{\partial}{\partial \eta}\left(\sum_{i=1}^\infty\frac{\pin_{>i}}{1+b_i(\eta)}\right) 
& = \sum_{i=1}^\infty \frac{\partial}{\partial \eta} \left(\frac{\pin_{>i}}{1+b_i(\eta)}\right)
=  -\sum_{i=1}^\infty \frac{\pin_{>i}(i^{-1}-(1-\beta))}{(1+b_i(\eta))^2}, \label{firstdiff}\\
\frac{\partial^2}{\partial \eta^2}\left(\sum_{i=1}^\infty\frac{\pin_{>i}}{1+b_i(\eta)}\right) 
& =\sum_{i=1}^\infty \frac{\partial^2}{\partial \eta^2} \left(\frac{\pin_{>i}}{1+b_i(\eta)}\right)
=  2\sum_{i=1}^\infty \frac{\pin_{>i}(i^{-1}-(1-\beta))^2}{(1+b_i(\eta))^3}.\label{secdiff}
\end{align}
It suffices to check:
\[
\sum_{i=1}^\infty\sup_{\eta\in \mathcal{I}}\left|\frac{\partial}{\partial \eta}\left(\frac{\pin_{>i}}{1+b_i(\eta)}\right)\right| <\infty,
\qquad 
\sum_{i=1}^\infty\sup_{\eta\in \mathcal{I}}\left|\frac{\partial^2}{\partial \eta^2}\left(\frac{\pin_{>i}}{1+b_i(\eta)}\right)\right| <\infty.
\]
Note that for $i\ge 1$,
\begin{align*}
\sup_{\eta\in\mathcal{I}}\left|\frac{\partial}{\partial \eta}\left(\frac{\pin_{>i}}{1+b_i(\eta)}\right) \right|
&= \sup_{\eta\in\mathcal{I}} \frac{\pin_{>i}|i^{-1}-(1-\beta)|}{(1+b_i(\eta))^2} \\
&\le (2-\beta) \sup_{\eta\in\mathcal{I}}\frac{\pin_{>i}}{(1+b_i(\eta))^2}
\le (2-\beta)(1+(1-\beta)K)^2\pin_{>i}.
\end{align*}
Recall \eqref{ipi_limit}, we then have
\[
\sum_{i=0}^\infty \pin_{>i} = \sum_{i=0}^\infty \sum_{k>i} \pin_k = \sum_{k=0}^\infty \sum_{i= 0}^{k-1} \pin_k
= \sum_{k=0}^\infty k\pin_k = 1.
\]
Hence, 
\begin{align*}
\sum_{i=1}^\infty  \sup_{\eta\in \mathcal{I}}\left|\frac{\partial}{\partial \eta}\left(\frac{\pin_{>i}}{1+b_i(\eta)}\right)\right| 
 &\le (2-\beta)(1+(1-\beta)K)^2 \sum_{i=0}^\infty\pin_{>i} \\
 &=  (2-\beta)(1+(1-\beta)K)^2<\infty,
\end{align*}
which implies \eqref{firstdiff}. Equation \eqref{secdiff} then follows by a similar argument.
Combining \eqref{tildeh}, \eqref{firstdiff} and \eqref{secdiff} gives
\begin{align*}
\frac{\partial^2}{\partial \eta^2}\tilde h(\eta) 
=  2&\left( \sum_{i=1}^\infty \frac{ \pin_{>i}}{1+b_i(\eta)}\right)^{-3} \\ 
\times \left[ \left( \sum_{i=1}^\infty \frac{\pin_{>i}(i^{-1}-(1-\beta))}{(1+b_i(\eta))^2}\right)^2 \right.& \left. -\left( \sum_{i=1}^\infty \frac{ \pin_{>i}}{1+b_i(\eta)}\right) \left(\sum_{i=1}^\infty \frac{\pin_{>i}(i^{-1}-(1-\beta))^2}{(1+b_i(\eta))^3}\right) \right] <0,
\end{align*}
by the Cauchy-Schwarz inequality.
 Hence $\tilde h$ is concave on $\mathcal{I}$. 

From Lemma~\ref{phi}, $\psi(\din)=0$ where $\psi(\cdot)$ is as defined in \eqref{defpsi}. Hence we have $f(\din)=\alpha + \beta$ in a similar derivation to that of \eqref{fndef}. Also from \eqref{alpha-p0}, we have $g(\din)=\alpha + \beta$. 
Hence, $\din$ is a solution to $f(\delta)=g(\delta)$. 

Under the $\delta\mapsto \eta$ reparametrization in \eqref{reparam}, we have that $\tilde f(\eta_\text{in})=\tilde g(\eta_\text{in})$ where $\eta_\text{in}:=\din/(1+\din(1-\beta))$, and also
$$
\lim_{\eta\downarrow 0} \tilde f(\eta) = \sum_{i=1}^\infty \pin_{>i} 
= 1- \pin_{>0} = \beta+\pin_0 = \lim_{\eta\downarrow 0} \tilde g(\eta).
$$
This, along with the concavity of $\tilde h$, implies that $\eta_\text{in}$ is the unique solution to $\tilde h(\eta)=0$, or equivalently, to $\tilde f(\eta)=\tilde g(\eta)$ on $\mathcal{I}$. 

Let $\tilde f_n(\eta):= f_n(\delta(\eta))$, $\tilde g_n(\eta):= g_n(\delta(\eta))$. We can show in a similar fashion that $\tilde\eta:= \tildedin/(1-\tildedin(1-\tilde\beta))$ is the unique solution to $\tilde f_n(\eta)=\tilde g_n(\eta)$. 
Using an analogue of the arguments in the proof of Theorem~\ref{unifconv}, we have 
$$ 
\sup_{\eta\in\mathcal{I}} |\tilde f_n(\eta)-\tilde f(\eta)|\convas 0, \quad \sup_{\eta\in\mathcal{I}} |\tilde g_n(\eta)-\tilde g(\eta)|\convas 0,
$$
and therefore $\tilde\eta \convas \eta_\text{in}$. Since $\delta \mapsto \eta$ is a one-to-one transformation from $[\epsilon,K]$ to $\mathcal{I}$, we have that $\tildedin$ is the unique solution to $f_n(\delta) = g_n(\delta)$ and that $\tildedin\convas\din$. On the other hand, $\tilde\alpha$ can be solved uniquely by plugging $\tildedin$ into \eqref{fndef} and is also strongly consistent, which completes the proof.

\end{proof}

\end{document}